\documentclass[10pt,a4paper]{article}

\usepackage{fullpage}
\usepackage{url} 
\usepackage{authblk}

\usepackage{amsfonts}
\usepackage{amsmath}
\usepackage{amssymb}
\usepackage{amsthm}
\usepackage{enumerate}
\usepackage{mdwlist}
\usepackage{mathrsfs}
\usepackage{bussproofs}
\usepackage{prooftree}
\usepackage{graphicx}
\usepackage{latexsym}

\newtheorem{theorem}{Theorem}[section]
\newtheorem{proposition}[theorem]{Proposition}
\newtheorem{lemma}[theorem]{Lemma}
\newtheorem{corollary}[theorem]{Corollary}

\theoremstyle{remark}
\newtheorem{definition}[theorem]{Definition}
\newtheorem{remark}[theorem]{Remark}
\newtheorem{example}[theorem]{Example}

\def\NN{\mathbb{N}}

\def\BB{\mathbb{B}}

\newcommand{\eqleft}[1]{\begin{itemize} \item[] $#1$ \end{itemize}}
\newcommand{\pair}[1]{\langle #1 \rangle}

\newcommand{\initSeg}[2]{[#1](#2)}
\newcommand{\initSegs}[3]{\{#2\}_{#1}(#3)}
\newcommand{\wext}[1]{\widehat{#1}}
\newcommand{\ext}[1]{\hat{#1}}

\newcommand{\modcont}{\mathcal{C}^\omega}
\newcommand{\modpar}{\hat{\mathcal{C}}^\omega}

\newcommand{\systemT}{{\sf T}}

\newcommand{\PA}{{\sf PA}}

\newcommand{\PAomega}{{\sf PA}^\omega}
\newcommand{\EHAomega}{{\sf E\mbox{-}HA}^\omega}
\newcommand{\EPAomega}{{\sf E\mbox{-}PA}^\omega}

\newcommand{\WEPAomega}{{\sf WE\mbox{-}PA}^\omega}

\newcommand{\CONT}{{\sf Cont}}
\newcommand{\sSPEC}{{\sf sSpec}}
\newcommand{\SPEC}{{\sf Spec}}

\newcommand{\QFAC}{{\sf QF\mbox{-}AC}}

\newcommand{\AC}{{\sf AC}}

\newcommand{\DC}{{\sf DC}}

\newcommand{\CA}{{\sf CA}}

\newcommand{\DNS}{{\sf DNS}}

\newcommand{\BI}{{\sf BI}}
\newcommand{\sBI}{{\sf sBI}}

\newcommand{\Rec}{{\sf R}}

\newcommand{\BR}{{\sf BR}}
\newcommand{\sBR}{{\sf sBR}}

\newcommand{\SBR}{{\Phi}}
\newcommand{\sSBR}{\Psi}

\newcommand{\types}{\mathcal{T}}
\newcommand{\compact}{\mbox{\sf compact}}
\newcommand{\discrete}{\mbox{\sf discrete}}
\newcommand{\zero}{{\bf 0}}

\newcommand{\dt}[3]{|{#1}|^{#2}_{#3}}

\newcommand{\least}{\mu}

\newcommand{\at}{{\; @ \:}}

\newcommand{\upd}[3]{{#1}\oplus({#2},{#3})}

\newcommand{\dom}{{\rm dom}}

\newcommand{\embpar}[1]{\mbox{emb}(#1)}

\newcommand{\unit}{{\bf 1}}

\begin{document}

\title{Spector Bar Recursion over Finite Partial Functions}

\author{Paulo Oliva and Thomas Powell}


\date{\today}

\maketitle

\begin{abstract}We introduce a new, demand-driven variant of Spector's bar recursion in the spirit of the Berardi-Bezem-Coquand functional of \cite{BBC(1998.0)}. The recursion takes place over finite partial functions $u$, where the control parameter $\varphi$, used in Spector's bar recursion to terminate the computation at sequences $s$ satisfying $\varphi(\ext{s})<|s|$, now acts as a guide for deciding exactly where to make bar recursive updates, terminating the computation whenever $\varphi(\ext{u})\in\dom(u)$. We begin by exploring theoretical aspects of this new form of recursion, then in the main part of the paper we show that demand-driven bar recursion can be directly used to give an alternative functional interpretation of classical countable choice. We provide a short case study as an illustration, in which we extract a new bar recursive program from the proof that there is no injection from $\NN\to\NN$ to $\NN$, and compare this to the program that would be obtained using Spector's original variant. We conclude by formally establishing that our new bar recursor is primitive recursively equivalent to the original Spector bar recursion, and thus defines the same class of functionals when added to G\"odel's system $\systemT$. \\



\end{abstract}

\section{Introduction}
\label{sec-intro}

In 1962 C. Spector extended G\"odel's functional, or \emph{Dialectica}, interpretation of classical arithmetic \cite{Goedel(1958.0)} to full classical analysis by proving that the functional interpretation of the negative translation of countable choice, and hence full arithmetical comprehension, could be realized by a novel form of recursion which has come to be known as \emph{Spector's bar recursion}  \cite{Spector(1962.0)}. Since then, this seminal work has been extended in several ways, and in particular a number of novel variants of bar recursion have been devised to give computational interpretations to classical analysis in new settings, to the extent that bar recursion, in one form or another, is one of the most recognisable methods of giving a computational interpretation to mathematical analysis.

Spector's original aim was to extend G{\"o}del's proof of the relative consistency of Peano arithmetic to classical analysis. For this purpose, bar recursion is very well suited, allowing us to elegantly and easily expand the soundness of the Dialectica interpretation to incorporate the double negation shift and thus classical countable choice. However, in recent decades applications of proof interpretations such as the Dialectica interpretation and modified realizability have moved away from foundational concerns and towards the more practical issue of extracting computational content from proofs. In line with this shift of emphasis comes an increasing interest in how the modes of computation assigned to non-constructive principles behave.

From this perspective, it could be argued that traditional bar recursion is not necessarily the best method of interpreting countable choice principles. The defining characteristic of Spector's bar recursion is that it carries out computations over some well-founded tree of finite sequences $s$, always making recursive calls in a sequential manner over extensions $s\ast \pair{x}$ of these finite sequences. This strict adherence to sequentiality means that in practice, when constructing an approximation to a choice sequence using bar recursion, elements of the approximation are always computed in order, even if we do not require knowledge of all of them.

The purpose of this paper is to introduce a new, demand-driven alternative to Spector's bar recursion, in which the order of the recursive calls is not fixed but rather directly controlled by its parameters. We first focus on recursion-theoretic issues, in particular giving an intuitive explanation of why our recursor exists in standard continuous models. We then prove that our new form of recursion is capable of realizing the Dialectica interpretation of countable choice, and moreover argue that (from an algorithmic viewpoint) it can be superior to Spector's original bar recursion because the manner in which it constructs approximations to choice sequences is much more sensitive to its environment. We illustrate this with an example in which we extract realizers from the classical proof that there is no injection from $\NN\to\NN$ to $\NN$. In this case the program based on our demand-driven recursion has a much more intuitive behaviour than that based on Spector bar recursion, and significantly outperforms the latter on a small sample of concrete inputs. Finally, we calibrate the computational strength of this new variant of bar recursion relative to Spector's original definition, showing that despite the algorithmic differences in extracted programs, the two forms of bar recursion are in fact primitive recursively interdefinable, and thus our recursor exists in all the usual models of Spector's bar recursion.

Our variant of bar recursion is in some ways similar to the realizer of countable choice proposed by Berardi et al. in \cite{BBC(1998.0)}, now often called the BBC-functional. In both cases the recursion is carried out in a `symmetric', rather than a fixed sequential manner. However, the BBC-functional belongs to the world of \emph{realizability}, which typically uses much stronger forms of recursion to interpret choice principles than Spector bar recursion (see \cite{BergOli(2006.0),EscOli(2015.0),Powell(2013.0)}). Moreover, the BBC-functional itself has a highly complex behaviour, its demand driven execution coming at the expense that each entry in its output is computed via a completely independent recursion. Our bar recursor is very different in this respect, as it retains a `memory' of what has already been computed, and simply relaxes the order in which the computation occurs.


\subsection{Preliminaries}
\label{subsec-intro-prelim}

The basic formal theory we work over is fully extensional\footnote{As it is well-known, the Dialectica interpretation does not validate full extensionallity. The system we are describing here, however, is the one we use to \emph{verify} the bar recursive interpretation of countable choice, and also for our inter-definability results.} Heyting arithmetic $\EHAomega$ in all finite types (and its classical variant $\EPAomega$), whose quantifier-free fragment is G{\"o}del's system $\systemT$ of primitive recursive functionals (see \cite{Kohlenbach(2008.0),Troelstra(1973.0)} for full details). The finite types $\types$ are typically defined using the following inductive rules
\begin{equation*}\types ::= \; \NN \; | \; X\times Y \; | \; X\to Y\end{equation*}
We expand these basic types with some standard `definable' types, including the unit $\unit$ and Boolean $\BB = \{0,1\}$ types, finite sequence types $X^\ast$ and co-product types $X+Y$. We consider \emph{partial sequences} $\NN\to X$ to be objects of type $\NN\to X+\unit$, where $\unit$ denotes an `undefined' value.

Finally, we consider one slightly non-standard type: the type $X^\dagger$ of \emph{finite partial functions}, that is partial sequences $\NN \to X$ defined at only finitely many points. This type can be easily simulated as in \cite{Berger(2002.0)} by $(\NN\times X)^\ast$, in which the partial function that takes values $x_0,\ldots,x_{k-1}$ at arguments $n_0,\ldots,n_{k-1}$ is encoded as the sequence $\pair{(n_0,x_0),\ldots,(n_{k-1},x_{k-1})}$, although to minimise syntax we treat $X^\dagger$ as primitive and avoid any further details of how it can be precisely encoded using the usual types. 

Relative to a suitable encoding for $X^\dagger$ there exists a computable functional $\dom \colon X^\dagger\to\NN^\ast$ which for each finite partial function $u$ encodes its finite domain as a sequence of natural numbers (if $u$ is encoded in the type $(\NN\times X)^\ast$ as described above this functional would simply be the first projection of the sequence, quotiented by equality). In general, for both partial and finite partial sequences $u$ we write $\dom(u) \subseteq \NN$ to denote the domain of $u$, and write $n\in\dom(u)$ whenever $u$ is defined at $n$. We can assume that membership of $\dom(u)$ is a decidable predicate (i.e. recursive in $u$). 

We write $x\colon X$ or $x^X$ to signify ``$x$ is an object of type $X$", and sometimes write $Y^X$ for the type $X\to Y$. In addition to the basic constructors for dealing with the finite types, $\EHAomega$ contains variables and quantifiers for all types, the predicates $=_\BB$ and $=_\NN$ and corresponding axioms for equality over base types, induction over arbitrary formulas, combinators which allow us to carry out $\lambda$-abstraction, and primitive recursors $\Rec_X$ for each type which satisfy
\begin{equation*}\begin{aligned}\Rec_X^{y,z}(0)&=_X y \\ \Rec_X^{y,z}(n+1)&=_X z_n(\Rec_X^{y,z}(n)).\end{aligned}\end{equation*}
In particular, the recursor of type $\NN$ allows us to carry out definition by cases, and also assign characteristic functions to all quantifier-free formulas of $\EHAomega$. Note that we choose to write $\Rec_X^{y,z}(n)$ instead of the more common $\Rec_X(y,z,n)$. This is simply a notational separation between the parameters that remain fixed throughout the recursion, namely $y$ and $z$, and the parameter over which the recursion takes place, namely $n$. We follow a similar convention when defining bar recursion schemes below.

Finally, higher-type equality $=_X$ for arbitrary $X$ is defined inductively in terms of $=_\NN$, and is treated as fully extensional via the axioms
\begin{equation*}\forall f^{X\to Y},x^X,y^X(x=_X y\to f(x)=_Y f(y)).\end{equation*}
We will often require extensions of $\EHAomega$/$\EPAomega$ with various principles, notably both the axioms $\AC_{\NN, X}$ of countable choice and $\DC$ of countable dependent, which are defined by
\begin{equation*}
\AC_{\NN, X} \; \colon \; \forall n^{\NN} \exists x^X A_n(x)\to\exists f^{\NN \to X}\forall n A_n(f(n))
\end{equation*}
and
\begin{equation*}
\DC_X \; \colon \; \forall n^{\NN},x^X \exists y^X A_n(x,y)\to\exists f^{\NN \to X}\forall n A_n(f(n),f(n+1))
\end{equation*}
respectively, where $A$ is an arbitrary formula in each case. We denote by $\AC_{\NN}$ and $\DC$ the general axiom schemata $(\AC_{\NN, X})$ and $(\DC_X)$ where $X$ ranges over arbitrary finite types.

\subsection{Notation}
\label{sec-intro-notation}

We make use of the following notational conventions:

\begin{itemize}

\item \emph{Finite sequence constructors}. For $s\colon X^\ast$, $|s|$ denotes the length of $s$. We use $\pair{x_0, \ldots, x_{n-1}} \colon X^*$ to denote the finite sequence constructor. Hence, we write $\pair{}$ for the empty sequence and $\pair{x}$ for the singleton sequence containing only $x$. 

\item \emph{Finite sequence concatenation}. Given two finite sequences $s, t \colon X^*$ we write $s * t$ for their concatenation. For a finite sequence $s\colon X^\ast$ and an object $x\colon X$ we often write $s\ast x$ for $s\ast\pair{x}$. We also write $s\ast\alpha\colon X^\NN$ to denote the concatenation of the finite sequence $s \colon X^*$ with an infinite sequence $\alpha\colon X^\NN$. 

\item \emph{Initial finite sequence}. Given an $\alpha \colon X^\NN$ we write $\initSeg{\alpha}{n}=\pair{\alpha(0),\ldots,\alpha(n-1)}$ for the finite initial segment of length $n$ of the infinite sequence $\alpha$.

\item \emph{Finite partial function constructors}. We use $\emptyset \colon X^\dagger$ for the finite partial function with empty domain, and $(n, x) \colon X^\dagger$ for the finite partial function defined only at point $n$ with value $x$.

\item \emph{Ordering of finite partial functions}. Given two finite partial functions $u, v \colon X^\dagger$ we write $u \sqsubseteq v$ if the domain of $v$ contains the domain of $u$; and $u$ and $v$ coincide on the domain of $u$. In the  domain of $u$ is strictly contained in the domain of $v$ we will write $u \sqsubset v$.

\item \emph{Merging finite partial functions}. Given two finite partial functions $u, v \colon X^\dagger$ we write $u \at v$ to denote the ``union" of the two partial functions, where we give priority to the values of $u$ when $u$ and $v$ are both defined at some common point. Given finite sequences $s, t \colon X^*$ we also write $s \at t$ since we can think of $s, t \colon X^*$ as finite partial functions of a particular kind.

\item \emph{Finite partial function update}. For a finite partial function $u \colon X^\dagger$ and an object $x\colon X$ we write $u \oplus (n,x)$ for $u \at (n, x)$, to stress the view that this is an update of $u$ with the new value $x$ at point $n$. Of course, if $u$ is already defined at point $n$ then $u \oplus (n,x) = u$.

\item \emph{Canonical extensions}. The term $\zero_X\colon X$ denotes the usual inductively defined zero object of each type $X$, used as a canonical representative of $X$ -- we use the convention that $\zero_{X^\ast}=\pair{}$ and $\zero_{X^\dagger}=\emptyset$. The canonical extension $\hat{s}\colon X^\NN$ of the finite sequence $s \colon X^*$ is defined by $\hat{s}(i)=s_i$ for $i<|s|$ and else $\hat{s}(i)=\zero_{X}$. The canonical extension $\hat{u} \colon X^\NN$ of a finite partial function $u \colon X^\dagger$ is defined analogously. Given a function $\varphi \colon X^\NN\to R$ we also talk about its canonical extension, and write $\hat{\varphi}$ for the function $\hat{\varphi}(s) = \varphi(\hat{s})$ so that $\hat{\varphi}$ can be either of type $X^* \to R$ or $X^\dagger \to R$. The type of $\hat{\varphi}$ will be clear from the context.

\item \emph{Partial application}. As a generalisation of currying, given a function $\varphi \colon X^\NN\to R$ and an $s \colon X^*$, we write $\varphi_s\colon X^\NN\to R$ for the function defined by $\varphi_s(\alpha) = \varphi(s\ast\alpha)$.

\item \emph{Bounded search}. Given any decidable predicate $P(i)$ on $\NN$, the term $\mu i\leq n\; . \; P(i)$ returns the least $i \leq n$ satisfying $P(i)$, if it exists, or returns $n$ otherwise.

\end{itemize}

\subsection{Spector's bar recursion}
\label{sec-bar-SBR}

The defining equation of Spector's general bar recursor $\BR_{X,R}$ is given by
\begin{equation*}
\BR_{X,R}^{\phi,b,\varphi}(s^{X^\ast}) =_R
\begin{cases}
b(s) & \mbox{if $\varphi(\hat{s})<|s|$} \\
\phi_s(\lambda x^{X} . \BR_{X,R}^{\phi,b,\varphi}(s\ast x)) & \mbox{otherwise}
\end{cases}
\end{equation*}
where the parameters have type $\phi\colon X^\ast\to (X\to R)\to R$, $b\colon X^\ast\to R$ and $\varphi\colon X^\NN \to \NN$ and $X$, $R$ range over arbitrary types. Just as we did with $\Rec_X^{y,z}(n)$, we also write $\BR_{X,R}^{\phi,b,\varphi}(s)$ instead of the more common $\BR_{X,R}(\phi,b,\varphi,s)$, so as to highlight the parameter $s$ over which the recursion takes place. In fact, we will often omit the parameters $\phi,b,\varphi$ from the superscript of $\BR$ when there is no danger of ambiguity. 

The parameter $\varphi$ acts as a `control' for $\BR^{\phi,b,\varphi}_{X,R}(s)$, whose role is to ensure that at some point the recursive calls stop. Therefore Spector's bar recursor is well-founded only if the control parameter eventually satisfies $\ext{\varphi}(s\ast\pair{x_0,\ldots,x_{N-1}})<|s|+N$ (recall that we use the abbreviation $\ext{\varphi}(t)=\varphi(\ext{t})$) for each sequence of recursive calls $s$, $s\ast \pair{x_0}$, $s\ast\pair{x_0,x_1}\dots$. We call this requirement \emph{Spector's condition}, which can be formulated more precisely as
\begin{equation*}\SPEC_X \ \colon \ \forall\varphi^{X^\NN\to\NN}\forall\alpha^{X^\NN}\exists n(\ext{\varphi}(\initSeg{\alpha}{n})<n).\end{equation*}
As demontrated by Howard using a trick attributed to Kreisel, $\SPEC$ must be valid in any model of bar recursion.
\begin{proposition}[Howard/Kreisel \cite{Howard(1968.0)}] \label{howard} $\EHAomega+(\BR)\vdash \SPEC$.\end{proposition}
For this reason, $\BR$ is not well-defined in the full type structure of all set-theoretic functionals, since $\SPEC$ is clearly not valid in this structure. However, it is well known to exist in most continuous type-structures (such as the Kleene/Kreisel continuous functionals \cite{Kleene(1959.0),Kreisel(1959.0),Troelstra(1973.0)}), and even in the type structure of strongly majorizable functionals \cite{Bezem(1985.0)}, which contains non-continuous functionals.  

Just as normal primitive recursion forms a computational analogue of induction, bar recursion can be viewed as a computational analogue of the principle of \emph{bar induction}, which is well-known in intuitionistic mathematics as an equivalent formulation of dependent choice:
\begin{equation*}\BI \ \colon \ \forall\alpha^{X^\NN}\exists n P(\initSeg{\alpha}{n})\wedge\forall t^{X^\ast}(\forall x^X P(t\ast x)\to P(t))\to P(\pair{}).\end{equation*}
Here $P$ is some predicate over finite sequences.
%
To see, on an infomal level, why bar recursion exists in continuous models, we first note that such models all satisfy the following sequential continuity principle:
\begin{equation*}\CONT \ \colon \ \forall \varphi^{X^\NN\to \NN},\alpha^{X^\NN}\exists N\forall \beta(\initSeg{\alpha}{N} =_{X^\ast} \initSeg{\beta}{N} \to \varphi(\alpha)=_{\NN}\varphi(\beta)). \end{equation*}
From this we can easily derive $\SPEC$: if $N$ is a point of continuity on $\varphi$ and $\alpha$ then $\ext{\varphi}(\initSeg{\alpha}{n}) < n$ holds for $n := \max\{N,\varphi(\alpha)+1\}$. Now, to show that $\BR_{X,R}^{\phi,b,\varphi}(s)$ defines a total object for total arguments $\phi$, $b$, $\varphi$ and $s$, we argue by bar induction on the predicate
$$P(t) \equiv \BR_{X,R}^{\phi,b,\varphi}(s\ast t)\mbox{ is total}.$$
Given an infinite sequence $\alpha\colon X^\NN$ it is clear by $\SPEC$ that $\ext{\varphi}(s\ast\initSeg{\alpha}{n})<|s|+n$ for some $n$ and therefore $\BR^{\phi,b,\varphi}(s\ast\initSeg{\alpha}{n})=b(s\ast\initSeg{\alpha}{n})$ is total. Clearly the bar induction step $\forall t^{X^\ast}(\forall x^X P(t\ast x)\to P(t))$ holds and thus we obtain $P(\pair{})$ and therefore totality of $\BR_{X,R}^{\phi,b,\varphi}(s)$. A broadly similar but somewhat more involved application of bar induction proves that $\BR$ exists in the majorizable functionals (see \cite{Bezem(1985.0),Kohlenbach(2008.0)}).

To summarise, the basic idea behind Spector's bar recursion is that any sequence of recursive calls made by $\BR$ eventually hits a \emph{bar} $s$ at which the condition $\varphi(\hat{s})<|s|$ holds and therefore $\BR(s)$ is assigned a value $b(s)$. These values propagate backwards along the tree of recursive calls ensuring that $\BR$ is defined everywhere.

\section{A demand-driven variant of Spector's bar recursion}
\label{sec-bar}

%

One can view Spector's bar recursion as just one instance of a more general form of backward recursion in which the main argument is some partial function with finite domain (for Spector a finite sequence $s$), and recursive calls are made by extending the domain of this argument (for Spector extending the sequence with one element $s\ast x$). From this perspective it seems that bar recursion is quite constrained in that the domain of its input is always an initial segment of the natural numbers. This has two obvious disadvantages. Firstly, the implicit dependence on the ordering of the natural numbers makes it unclear how to generalise $\BR$ to carry out recursion over partial functions on discrete structures which do not come equipped with a natural ordering. Secondly, adherence to sequentiality means that precise values of the control functional $\varphi$ are never required: all that matters is whether or not $\varphi(\hat{s})<|s|$, or in other words, whether or not $\varphi(\hat{s})$ is within the domain of already computed values. Thus when we consider that in terms of program extraction the parameter $\varphi$ is typically some realizing term extracted from a lemma in a proof, Spector's bar recursor lacks sensitivity in that it ignores precise information from its proof-theoretic environment.

It is natural, then, to ask whether there is alternative to bar recursion which still terminates on inputs $u$ with $\varphi(\ext{u})$ in the domain of $u$, but which searches for these points in a more flexible way, taking into account information provided by $\varphi$. This is the idea behind our variant of bar recursion, which we call \emph{symmetric bar recursion}. The symmetric bar recursor $\sBR_{X,R}$ is given by the defining equation 
\begin{equation*}
\sBR_{X,R}^{\phi,b,\varphi}(u^{X^\dagger})=_R
\begin{cases}
	b(u) & \mbox{if $\ext\varphi(u)\in\dom(u)$} \\
	\phi_u(\lambda x^X\; . \; \sBR^{\phi,b,\varphi}(\upd{u}{\ext\varphi(u)}{x})) & \mbox{otherwise}
\end{cases}
\end{equation*}
where now the parameters have type $\phi\colon X^\dagger \to (X\to R)\to R$, $b\colon X^\dagger\to R$ and $\varphi\colon X^\NN\to \NN$. Recall that the operation $\oplus$ indicates the extension of the partial function $u$ with one more piece of information, analogous to the extension of finite partial functions in the defining equation of $\BR$. 
The crucial difference is that this extension can potentially take place at any point $n\in\NN \backslash \dom(u)$, and so we are no longer restricted to making recursive calls in a sequential fashion. However, this additional freedom requires us to carefully justify the totality of $\sBR$, as its recursive calls are not easily seen to be well-founded. In Definition \ref{def-sBI} below we give a corresponding symmetric bar induction principle which can be used to reason about $\sBR$. First we need the following important definition.

\begin{definition}[Finite $\varphi$-threads of $u \colon X^\dagger$ or $\alpha \colon X^\NN$] \label{defn-S} Given $\varphi\colon X^\NN\to\NN$ and $u\colon X^\dagger$, the \emph{$\varphi$-thread of $u$ of length $i$} is the finite partial function $\initSegs{\varphi}{u}{i} \colon X^\dagger$ inductively defined as
\begin{equation*}
\begin{aligned}
\initSegs{\varphi}{u}{0}&:=\emptyset \\
\initSegs{\varphi}{u}{i+1}&:= \begin{cases} \upd{\initSegs{\varphi}{u}{i}}{n_{\varphi,i}}{u(n_{\varphi,i})} & \mbox{if $n_{\varphi,i}\in\dom(u)$}\\  \initSegs{\varphi}{u}{i} & \mbox{otherwise}\end{cases}
\end{aligned}
\end{equation*}
where $n_{\varphi,i}:=\ext{\varphi}(\initSegs{\varphi}{u}{i})$. Note that when either $n_{\varphi,i}\in\dom(\initSegs{\varphi}{u}{i})$ or $n_{\varphi,i}\notin\dom(u)$, we would have $\initSegs{\varphi}{u}{j}=\initSegs{\varphi}{u}{i}$ for all $j\geq i$. Entirely analogously, define the \emph{$\varphi$-thread of $\alpha$ of length $i$}, also denoted $\initSegs{\varphi}{\alpha}{i} \colon X^\dagger$, as
\begin{equation*}
\begin{aligned}
\initSegs{\varphi}{\alpha}{0}&:=\emptyset \\
\initSegs{\varphi}{\alpha}{i+1}&:= \upd{\initSegs{\varphi}{\alpha}{i}}{n_{\varphi,i}}{\alpha(n_{\varphi,i})}
\end{aligned}
\end{equation*} 
where $n_{\varphi,i}:=\ext{\varphi}(\initSegs{\varphi}{\alpha}{i})$.
%
%
%
\end{definition}

\begin{remark}In what follows we will frequently just write $\initSegs{}{u}{i}$ when $\varphi$ is clear from the context.\end{remark}

\begin{definition}[$\varphi$-threads] \label{defn-Sprop} Let $\varphi\colon X^\NN\to\NN$. We say that a finite partial function $u \colon X^\dagger$ is a \emph{$\varphi$-thread} if $u = \initSegs{\varphi}{u}{|\dom(u)|}$. This can be expressed formally by the decidable predicate 
\begin{equation*}
S_\varphi(u):\equiv\forall n(n\in\dom(u)\to n\in\dom(\initSegs{\varphi}{u}{|\dom(u)|})).
\end{equation*}
It is an easy exercise to see that $S_\varphi(u)\Leftrightarrow u=\initSegs{\varphi}{u}{|\dom(u)|}$, since for all $i$, $n\in\dom(\initSegs{\varphi}{u}{i})$ implies that $\initSegs{\varphi}{u}{i}(n)=u(n)$.\end{definition}

The intuition here is that $\varphi$ works as a control function that dictates the location of the next bar recursive call. Thus $\varphi$-threads are just partial functions that have been constructed using $\varphi$ as a control functional.

\begin{lemma}\label{lem-S}Suppose that $n_{\varphi,i}$ is defined as in Definition \ref{defn-S}. A finite partial function $u$ satisfies $S_{\varphi}(u)$ (i.e. is a $\varphi$-thread) iff for all $i\leq |\dom(u)|$,
\begin{equation*}\initSegs{\varphi}{u}{i}=(n_{\varphi,0},x_0)\oplus (n_{\varphi,1},x_1)\oplus\ldots\oplus (n_{\varphi,i-1},x_{i-1})\end{equation*}
for $x_j = u(n_{\varphi, j})$, where the $n_{\varphi,j}$ are all distinct and all lie in $\dom(u)$. In particular
\begin{equation*}S_{\varphi}(u)\Rightarrow u=\initSegs{\varphi}{u}{l}=(n_{\varphi,0},x_0)\oplus\ldots\oplus (n_{\varphi,l-1},x_{l-1})\end{equation*} 
where $l=|\dom(u)|$.\end{lemma}

\begin{proof}For one direction, assume $S_\varphi(u)$ and set $l:=|\dom(u)|$. We use induction on $i \leq l$. If $i = 0$ then $\initSegs{\varphi}{u}{0}=\emptyset$ by definition. Now for $i<l$ assume that $\initSegs{\varphi}{u}{i}=(n_{\varphi,0},x_0)\oplus\ldots\oplus (n_{\varphi,i-1},x_{i-1})$ for distinct $n_{\varphi,j}$. If either $n_{\varphi, i} \not\in \dom(u)$ or $n_{\varphi, i} \in \{n_{\varphi, 0}, \ldots, n_{\varphi, i-1}\}$ then, by the definition of the $\varphi$-thread of $u$, we would have that $\initSegs{\varphi}{u}{i} = \initSegs{\varphi}{u}{i+1} = \ldots = \initSegs{\varphi}{u}{l}$, and hence $i=|\dom(\initSegs{\varphi}{u}{i})|=|\dom(\initSegs{\varphi}{u}{l})| < l$, contradicting $S_\varphi(u)$. Therefore, we must have $n_{\varphi,i}\in\dom(u)\backslash\{n_{\varphi,0},\ldots,n_{\varphi,i-1}\}$ and
\begin{equation*}\initSegs{\varphi}{u}{i+1}=(n_{\varphi,0},x_0)\oplus\ldots\oplus (n_{\varphi,i-1},x_{i-1})\oplus(n_{\varphi,i},x_i)\end{equation*}
for $x_i=u(n_{\varphi,i})$, where the $n_{\varphi,j}$ are all distinct and belong to $\dom(u)$. The other direction is straightforward: If $\initSegs{\varphi}{u}{l} = (n_{\varphi,0},x_0)\oplus\ldots\oplus (n_{\varphi,l-1},x_{l-1})$ for distinct $n_{\varphi,j}$ then $\dom(\initSegs{\varphi}{u}{l}) = l = \dom(u)$. But since $\initSegs{\varphi}{u}{l}$ is only defined at points where $u$ is defined, and at those points they hold the same value, it means that $u = \initSegs{\varphi}{u}{l}$, i.e. $S_\varphi(u)$.\end{proof}

\begin{example}Define $u:=(1,1)\oplus (2,2)\oplus (3,3)$ i.e. $u$ is the partial identity function defined at $1,2,3$.
\begin{enumerate}

\item Let $\varphi(\alpha):=\max\{\alpha(0),\alpha(1),\alpha(2)\}+1$. We have that $n_{\varphi,0} = \ext{\varphi}(\emptyset) = \max\{0,0,0\}+1=1$ and thus $\initSegs{\varphi}{u}{1}=(1,1)$, and similarly $n_{\varphi,1}=\max\{0,1,0\}+1=2$, $\initSegs{\varphi}{u}{2}=(1,1)\oplus (2,2)$, and $n_{\varphi,2}=\max\{0,1,2\}+1=3$, $\initSegs{\varphi}{u}{3}=(1,1)\oplus (2,2)\oplus (3,3)=u$. Hence, $u$ is a $\varphi$-thread.

\item On the other hand, for $\psi(\alpha):=\max\{\alpha(0),\alpha(1),\alpha(2)\}$ we have $n_{\psi,0}=0\notin\dom(u)$ and thus $\initSegs{\psi}{u}{i}=\emptyset$ for all $i$, so $u$ is not a $\psi$-thread.


\end{enumerate}
\end{example}

\begin{definition} Let us write $\forall u \!\in\! S_{\varphi} \, A(u)$ as an abbreviation for $\forall u (S_{\varphi}(u) \to A(u))$. The principle of \emph{symmetric bar induction} $\sBI$ is given by
\begin{equation*}\label{def-sBI}
\sBI \ \colon \ \forall\varphi^{X^\NN\to\NN}\left(\forall\alpha^{X^\NN}\exists n P(\initSegs{\varphi}{\alpha}{n}) \wedge \forall u\in S_\varphi([\varphi(\ext{u})\notin\dom(u)\to \forall x^X P(\upd{u}{\varphi(\ext{u})}{x})]\to P(u))\to P(\emptyset)\right)
\end{equation*}
where $P$ is an arbitrary predicate on $X^\dagger$. 
%
%
\end{definition}
\begin{theorem} \label{thm-sBI} $\EPAomega + \DC \vdash \sBI$.
\end{theorem}

\begin{proof}Fix some $\varphi$ and suppose for a contradiction that the premises of $\sBI$ are true but $\neg P(\emptyset)$. The second premise of $\sBI$ is classically equivalent to
\[ \forall u\in S_\varphi(\neg P(u) \to [\varphi(\ext{u})\notin\dom(u) \wedge \exists x \neg P(\upd{u}{\varphi(\ext{u})}{x})]). \]
Hence, by dependent choice, there exists a sequence $u_0,u_1,\ldots$ of finite partial functions, together with a sequence $x_0,x_1,\ldots$ of elements of $X$, satisfying
\begin{equation*}u_0=\emptyset \mbox{ \ \ \ and \ \ \ } u_{i+1}=u_i\oplus (n_{i},x_i)\end{equation*}
with $n_i = \ext{\varphi}(u_i)\notin\dom(u_i)$. Clearly each $u_i$ is a $\varphi$-thread, i.e. $S_{\varphi}(u_i)$. Moreover,  by construction we have that $\neg P(u_i)$  holds for all $i$. Now, by classical countable choice there exists a function $\alpha\colon X^\NN$ satisfying
\begin{equation*}
\alpha(n):=
\begin{cases}
u_{i}(n) & \mbox{where $i$ is the least such that $n\in\dom(u_i)$, if it exists} \\
\zero_X & \mbox{otherwise}. 
\end{cases}
\end{equation*}
Let us first show by induction on $i$ that $\initSegs{\varphi}{\alpha}{i}=u_i$, for all $i$. First, $\initSegs{\varphi}{\alpha}{0}=\emptyset$ by definition. Assuming that $\initSegs{\varphi}{\alpha}{i}=u_i$ we have $\ext{\varphi}(\initSegs{\varphi}{\alpha}{i}) = \ext{\varphi}(u_i)=n_i$, and therefore $\initSegs{\varphi}{\alpha}{i+1}=\initSegs{\varphi}{\alpha}{i}\oplus (n_i,\alpha(n_i))=u_i\oplus (n_i,\alpha(n_i))$. Now by construction $n_i\notin\dom(u_i)$ and $n_i\in\dom(u_{i+1})$. Thus $\alpha(n_i)=u_{i+1}(n_i)=x_i$ and therefore $\initSegs{\varphi}{\alpha}{i+1}=u_{i+1}$. That concludes the proof that $\initSegs{\varphi}{\alpha}{i}=u_i$. \\
By the first premise of $\sBI$ there exists some $n$ such that $P(\initSegs{\varphi}{\alpha}{n})$, which implies $P(u_n)$, contradicting the assumption that $\neg P(u_i)$ holds for all $i$. \end{proof}

Intuitively, symmetric bar recursion is well-founded only if every sequence of recursive calls eventually arrives at some $u$ satisfying $\ext{\varphi}(u)\in\dom(u)$. Put formally, this statement can be seen as a symmetric analogue of $\SPEC$, namely 
\begin{equation*}
\sSPEC_X \ \colon \ \forall \varphi^{X^\NN\to\NN} \forall \alpha^{X^\NN} \exists n (\ext{\varphi}(\initSegs{\varphi}{\alpha}{n}) \in \dom(\initSegs{\varphi}{\alpha}{n})).
\end{equation*}
In fact, by adapting the proof in \cite{Howard(1968.0)} of Proposition \ref{howard}, we can prove that $\sSPEC$ must be valid in any model of $\sBR$.

\begin{proposition} \label{prop-sbar} Define the term $\theta_{\varphi,\alpha}$ in $\EHAomega + (\sBR)$ with free variables $\alpha\colon X^\NN$ and $\varphi\colon X^\NN\to\NN$ by
\begin{equation*}
\theta_{\varphi,\alpha}(u^{X^\dagger})=
	\begin{cases}
		0 & \mbox{if $\varphi(\ext{u})\in\dom(u)$} \\
		1+\theta_{\varphi,\alpha}(u \oplus (\varphi(\ext{u}),\alpha(\varphi(\ext{u})))) & \mbox{otherwise}.
	\end{cases}
\end{equation*}
Then, provably in $\EHAomega + (\sBR)$, we have $\ext{\varphi}(\initSegs{\varphi}{\alpha}{n})\in \dom(\initSegs{\varphi}{\alpha}{n})$ for some $n \leq \theta_{\varphi,\alpha}(\emptyset)$.
\end{proposition}

\begin{proof} Fix $\alpha \colon X^\NN$ and $\varphi \colon X^\NN \to \NN$. Let $\beta i := \theta_{\varphi,\alpha}(\initSegs{\varphi}{\alpha}{i})$. By definition of $\theta_{\varphi,\alpha}$ we have
\begin{equation*}
\beta i=
	\begin{cases}
		0 & \mbox{if $\ext{\varphi}(\initSegs{\varphi}{\alpha}{i}) \in \dom(\initSegs{\varphi}{\alpha}{i})$} \\
		1+\beta(i+1) & \mbox{otherwise}.
	\end{cases}
\end{equation*}
First note that, by the definition of $\beta$, we have 
\begin{equation*}(*) \ \ \beta i \neq 0\mbox{ \; iff \; }\ext{\varphi}(\initSegs{\varphi}{\alpha}{i}) \not\in \dom(\initSegs{\varphi}{\alpha}{i}).\end{equation*} 
By induction on $i$, using $(*)$, it is easy to show
\[ \forall i (\forall j \leq i \, (\beta j \neq 0) \to \forall j \leq i \, (\beta j=1+\beta(j+1))). \]
By another induction on $i$, using the above fact, we obtain
\[ \forall i (\forall j < i \, (\beta j \neq 0) \to \beta 0=i+\beta i). \]
Therefore, setting $i=\beta 0$ we have 
\begin{equation*}\forall j<\beta 0 (\beta j \neq 0) \to \beta 0= \beta 0 +\beta (\beta 0).\end{equation*} 
Therefore either $\beta(\beta 0)=0$, or $\exists j<\beta 0(\beta j = 0)$, i.e. $\exists j \leq \beta 0(\beta j = 0)$. Using $(*)$ we have 
\[ \exists j \leq \beta 0(\ext{\varphi}(\initSegs{\varphi}{\alpha}{j})\in\dom(\initSegs{\varphi}{\alpha}{j})). \]
That concludes the proof since $\beta 0 =\theta_{\varphi,\alpha}(\initSegs{\varphi}{\alpha}{0})=\theta_{\varphi,\alpha}(\emptyset)$. \end{proof}

\subsection{Relating $\SPEC$ and $\sSPEC$}

We now make our first link between symmetric bar recursion and Spector's bar recursion via their corresponding axioms $\sSPEC$ and $\SPEC$.

\begin{theorem} \label{thm-sSPEC-SPEC} $\EHAomega + \sSPEC_{X \times \BB} \vdash \SPEC_X$
\end{theorem}
\begin{proof} Given $\alpha\colon X^\NN$  and $\varphi\colon X^\NN\to\NN$ we need to produce a point $n$ such that $\ext{\varphi}(\initSeg{\alpha}{n}) < n$. Recall that $\BB = \{0, 1\}$ is the type of Booleans. Define $\tilde\alpha\colon (X\times\BB)^\NN$ and $\theta \colon (X\times\BB)^\NN\to\NN$ in terms of $\alpha$ and $\varphi$ as
\begin{itemize}
	\item[] $\tilde\alpha(n):=\pair{\alpha(n),1}$,
	\item[] $\theta(\beta):=\least i\leq \varphi(\lambda k . \pi_0(\beta k))(\pi_1(\beta i)=_\BB 0)$
\end{itemize}
where $\pi_0 \colon X \times \BB \to X$ and $\pi_1 \colon X \times \BB \to \BB$ are the two projections, and $\least$ is the bounded search operator. Intuitively, we are using the booleans to indicate whether a position is `defined' (i.e. equal to $1$) or not. Hence, the functional $\theta$ returns the first undefined position of $\beta$ which is bounded by $\varphi(\lambda k . \pi_0(\beta k))$, or just $0$ if no such position is found. By $\sSPEC$ there exists some $N$ such that
\begin{itemize}
	\item[($i$)] $\ext{\theta}(\initSegs{\theta}{\tilde{\alpha}}{N})\in\dom(\initSegs{\theta}{\tilde{\alpha}}{N})$.
\end{itemize}
Without loss of generality let $N$ be the least such value. We will show that $\ext{\varphi}(\initSeg{\alpha}{N}) < N$. To do this, we first claim that
\begin{itemize}
	\item[($ii$)] $\forall m \leq N (\dom(\initSegs{\theta}{\tilde{\alpha}}{m}) = \{0, \ldots, m-1\})$ and $\forall m < N (\ext{\varphi}(\initSeg{\alpha}{m})\geq m)$.
\end{itemize}	
The proof of ($ii$) is by induction on $m$. If $m = 0$ the claim is trivial. Now assume that ($ii$) holds for $m$. \\
For the first part, suppose that $m < N$. Then by the induction hypothesis we have
\begin{equation*}
\begin{aligned}
\ext{\theta}(\initSegs{\theta}{\tilde\alpha}{m})
	&=\theta(\pair{\alpha(0),1},\ldots,\pair{\alpha(m-1),1},\pair{\zero,0},\pair{\zero,0},\ldots)\\
	&=\least i\leq \ext{\varphi}(\initSeg{\alpha}{m})(i\geq m)\\
	&=m
\end{aligned},
\end{equation*}
using that $m < N$ and thus $\ext{\varphi}(\initSeg{\alpha}{m})\geq m$ by the second induction hypothesis. Therefore by definition we have $\initSegs{\theta}{\tilde\alpha}{m+1} = \initSegs{\theta}{\tilde\alpha}{m}\oplus (m,\tilde\alpha(m))$, and thus $\dom(\initSegs{\theta}{\tilde\alpha}{m+1})=\{0,\ldots,m\}$. \\
For the second part, suppose for a contradiction that $m<N$ but $\ext{\varphi}(\initSeg{\alpha}{m})<m$. Then by the first part we have
\begin{equation*}\ext{\theta}(\initSegs{\theta}{\tilde\alpha}{m})\leq \ext\varphi(\initSeg{\alpha}{m})<m\end{equation*}
which would imply that $\ext{\theta}(\initSegs{\theta}{\tilde\alpha}{m})\in\{0,\ldots,m-1\}$, contradicting the assumed minimality of $N$. \\
Therefore we have established ($ii$), and setting $m=N$ we have $\dom(\initSegs{\theta}{\tilde\alpha}{N})=\{0,\ldots,N-1\}$. By $(i)$, this of course implies that $\ext\theta(\initSegs{\theta}{\tilde\alpha}{N})<N$, and unwinding the definition $\theta$ we obtain
\begin{equation*}
N > \ext\theta(\initSegs{\theta}{\tilde\alpha}{N})=\least i\leq \ext{\varphi}(\initSeg{\alpha}{N})(i\geq N).
\end{equation*}
If $\ext{\varphi}(\initSeg{\alpha}{N})\geq N$ then the unbounded search would select $N$. So we must have that $\ext{\varphi}(\initSeg{\alpha}{N})< N$. \end{proof}

%
%
%

\begin{theorem} \label{thm-SPEC-sSPEC} $\EPAomega + \AC_{\NN, X} + \SPEC_{X^\dagger} \vdash \sSPEC_X$.
\end{theorem}
\begin{proof} Let $\alpha\colon X^\NN$ and $\varphi\colon X^\NN\to\NN$ be given. We must find $n$ such that $\ext{\varphi}(\initSegs{\varphi}{\alpha}{n}) \in \dom(\initSegs{\varphi}{\alpha}{n})$. Using $\AC_{\NN, X}$ and classical logic we can define the sequence of indices $(i_n)_{n \in \NN}$ as
\eqleft{
i_n :=
\begin{cases}
	i & \mbox{where $i$ is the least such that $n \in \dom(\initSegs{\varphi}{\alpha}{i})$, if such $i$ exists} \\[1mm]
	0 & \mbox{if no such $i$ exists}.
\end{cases}
}
Using $(i_n)_{n \in \NN}$ we can then define the sequence $\tilde\alpha\colon (X^\dagger)^\NN$ as
\eqleft{
\tilde\alpha(n):=\initSegs{\varphi}{\alpha}{i_n.} 
}
Note that $\tilde\alpha$ represents a characteristic function in the following sense:
\begin{equation*}
(\ast) \ \ n \in \dom(\tilde\alpha(n)) \Leftrightarrow \exists i(n\in\dom(\initSegs{\varphi}{\alpha}{i}))
\end{equation*}
Next, we primitive recursively define the `diagonalisation' function $d\colon (X^\dagger)^\NN\to (X+\unit)^\NN$ by
\eqleft{
d(\beta)(j)
:=_{X+\unit}
\begin{cases}
	\beta(k)(j) & \mbox{for least $k\leq j$ such that $j\in\dom(\beta(k))$} \\[1mm]
	\mbox{undefined} & \mbox{if no such $k \leq j$ exists}.
\end{cases}
}
Finally, we define the functional $\theta\colon (X^\dagger)^\NN\to\NN$ by 
\eqleft{\theta(\beta):=\ext{\varphi}(d(\beta)).}
Now, applying $\SPEC_{X^\dagger}$ to $\theta$ and $\tilde\alpha$ obtain a number $N$ such that
\begin{itemize}
\item[($i$)] $\ext{\theta}(\initSeg{\tilde\alpha}{N})<N$.
\end{itemize}
Let $i_m$ be the maximum number in the set $\{i_0,\ldots,i_{N-1}\}$. We claim that $i_m$ is our desired witness, i.e. $\ext{\varphi}(\initSegs{\varphi}{\alpha}{i_m})\in\dom(\initSegs{\varphi}{\alpha}{i_m})$. First, let $\embpar{\cdot} \colon X^\dagger \to (X+\unit)^\NN$ denote the embedding of finite partial functions into the type of arbitrary partial sequences. We prove that
\begin{itemize}
\item[($ii$)] $\ext{d}(\initSeg{\tilde\alpha}{N})=\embpar{\tilde\alpha(m)}$.
\end{itemize}
We consider two cases: \\
If $j\notin \dom(\ext{d}(\initSeg{\tilde\alpha}{N}))$ then by definition of $d$ we have that $\neg \exists k \leq j (j \in \dom(\widehat{\initSeg{\tilde\alpha}{N}}(k)))$. We consider a further two cases. If $j<N$ then $j\notin\dom(\tilde\alpha(k))$ for all $k\leq j$, which in particular implies $j\notin\dom(\tilde\alpha(j))$ and thus by $(\ast)$ we get $j\notin\dom(\initSegs{\varphi}{\alpha}{i_m})=\tilde\alpha(m)$. If $j \geq N$ then $j\notin\dom(\tilde\alpha(k))$ for all $k<N$, which in particular implies $j\notin\dom(\tilde\alpha(m))$. \\
If $j\in\dom(d(\widehat{\initSeg{\tilde\alpha}{N}}))$ we know that $j\in\dom(\tilde\alpha(k))$ for some $k\leq j$ (with $k<N$) and $\ext{d}(\initSeg{\tilde\alpha}{N})(j)=\tilde\alpha(k)(j)=\alpha(j)$, and since $i_k\leq i_m$ we have $\tilde\alpha(k) \sqsubseteq \tilde\alpha(m)$ and thus $\tilde\alpha(m)(j)=\alpha(j)$, which concludes the proof of ($ii$). \\
Therefore, using our usual abbreviation $n_{\varphi,k}=\ext\varphi(\initSegs{\varphi}{\alpha}{k})$, we obtain that $n_{\varphi, i_m} < N$ as follows:
\[
n_{\varphi,i_m} = \ext{\varphi}(\initSegs{\varphi}{\alpha}{i_m})=\ext{\varphi}(\embpar{\initSegs{\varphi}{\alpha}{i_m}})=\ext{\varphi}(\embpar{\tilde\alpha(m)})\stackrel{(ii)}{=}\ext{\varphi}(d(\widehat{\initSeg{\tilde\alpha}{N}}))=\theta(\widehat{\initSeg{\tilde\alpha}{N}})\stackrel{(i)}{<}N.
\]
Now, to prove the main result, suppose for a contradiction that $n_{\varphi,i_m}\notin\dom(\initSegs{\varphi}{\alpha}{i_m})$. Then by the definition of $\initSegs{\varphi}{\alpha}{i}$ we have $n_{\varphi,i_m}\in\dom(\initSegs{\varphi}{\alpha}{i_m+1})$, and moreover $i_m+1$ is the least such index and we have $i_{n_{\varphi,i_m}}=i_m+1$. But by maximality of $i_m$ for the set $\{i_0,\ldots,i_{N-1}\}$ and the fact that $n_{\varphi,i_m}<N$ we have $i_m\geq i_{n_{\varphi,i_m}}=i_m+1$, a contradiction.\end{proof}

%
%
%
%

In section \ref{sec-equiv} we expand the ideas in the proofs of Theorems \ref{thm-sSPEC-SPEC} and \ref{thm-SPEC-sSPEC} to show that the recursors $\BR$ and $\sBR$ themselves are primitive recursively equivalent. This in particular implies the following:

\begin{theorem}\label{thm-sBR-CONT}The Kleene/Kreisel continuous functionals $\modcont$ are a model of $\sBR$.\end{theorem}

\begin{proof} By Theorem \ref{thm-BR-sBR}, there is a term in $\mathcal{T} :\equiv \EHAomega + \sBI + \BR$ which satisfies the defining equation of $\sBR$ provably in $\mathcal{T}$. By Theorem \ref{thm-sBI} we can reduce $\sBI$ to $\DC$ over $\EPAomega$, and thus $\mathcal{T}$ can be reduced to $\EPAomega+\DC+\BR$. But it is well-known (see e.g. \cite{Troelstra(1973.0)}) that $\modcont$ is a model of $\EPAomega+\DC+\BR$, and therefore we can conclude that $\sBR$ exists in $\modcont$.  \end{proof}

We can also give a more direct (though informal) domain-theoretic argument to the existence of $\sBR$ in $\modcont$. 
%
%
It is a standard result \cite{Ershov(1977.0)} that the total continuous functionals $\modcont$ are just the extensional collapse of the partial continuous functionals $\modpar$. Now $\sBR$ can be easily defined in $\modpar$ as the least fixpoint $\Phi$ of its recursive defining equation, since $\modpar$ has the property that all recursive functionals $Z\to Z$ have a fixed point $p\colon Z$.

Now it is well-known that all total continuous functionals $\varphi$, $\alpha$ satisfy Spector's property $\SPEC$, and thus adapting Theorem \ref{thm-SPEC-sSPEC} to the total continuous functionals using the fact that $\AC_{\NN}$ is validated in continuous models (since these contain all set-theoretic sequences $\NN\to X$), we have that $\sSPEC$ also holds for all total $\varphi$, $\alpha$.

We can then prove that the fixpoint $\Phi$ defined is a total function using $\sBI$ (also valid in continuous models since as shown in Theorem \ref{thm-sBI} it follows directly from $\DC$). Suppose that the parameters $\phi$, $\varphi$ and $b$ and an argument $v$ of $\Phi$ are total, and let
\begin{equation*}
P(u^{X^\dagger}):\equiv \mbox{$\Phi(v\at u)$ is total}.
\end{equation*}
Then by $\sSPEC$ on the functional $\psi(\alpha):= \varphi(v\at\alpha)$ we have that $\forall \alpha \exists n (\ext{\psi}(\initSegs{\psi}{\alpha}{n}) \in \dom(\initSegs{\psi}{\alpha}{n}))$. By the definition of $\psi$ this implies
$$\forall \alpha \exists n (\ext{\varphi}(v \at \initSegs{\varphi}{\alpha}{n}) \in \dom(v \at \initSegs{\varphi}{\alpha}{n}))$$
and hence $\forall\alpha\exists n P(\initSegs{\varphi}{\alpha}{n})$, the first hypothesis of $\sBI$. For the second hypothesis of $\sBI$ as follows, we assume that for $u\in S_{\psi}$ we have $\psi(\ext{u}) \notin \dom(u) \to \forall x^X P(\upd{u}{\psi(\ext{u})}{x})$, and aim to prove that $\Phi(v \at u)$ is total. If $\Phi(v \at u)$ was not total then obviously we would have that $\ext{\psi}(u)\notin\dom(v\at u)$, which implies $\ext{\psi}(u)\notin\dom(u)$. But in this case $\forall x P(u\oplus (\ext{\psi
}(u),x)))$ implies that $\lambda x . \Phi(v\at u\oplus (\ext{\varphi}(v\at u),x))$ is total, and thus so is $\Phi(v\at u)$. Therefore since both premises of $\sBI$ hold, we can conclude $P(\emptyset)$ i.e. $\Phi(u)$ is total. 

\section{The Dialectica interpretation of countable choice}
\label{sec-dialectica}

In the following sections we assume that the reader is broadly familiar with G\"{o}del's Dialectica intepretation and its role in the extraction of computational content from proofs -- details of which can be found in e.g. \cite{AvFef(1998.0),Kohlenbach(2008.0)} -- although we make an effort to keep the main flow of ideas as self-contained as possible. 

Recall that the Dialectica interpretation translates each formula $A$ in the language of some theory $\mathcal{T}$ to a quantifier-free formula $\dt{A}{x}{y}$ in some `verifying' functional theory $\mathcal{S}$, where $x$ and $y$ are (possibly empty) tuples of variables of some finite type. The idea is that $A$ is (classically) equivalent to $\exists x\forall y\dt{A}{x}{y}$, and that the interpretation $\mathcal{T}\to\mathcal{S}$ is sound if whenever $\mathcal{T}\vdash A$ we can extract a realizer $t$ for $\exists x$ so that $\mathcal{S}\vdash \dt{A}{t}{y}$. When $\mathcal{T}$ is a classical theory, one typically precomposes the Dialectica intepretation with a negative translation in order to obtain soundness, a combination normally referred to as the \emph{ND interpretation}.

The Dialectica interpretation was conceived by G\"{o}del in the 1930s, and published much later in a seminal paper of 1958 \cite{Goedel(1958.0)} in which it was shown that Peano arithmetic can be ND interpreted into the system $\systemT$ of primitive recursive functionals in all finite types. In fact it is not too difficult to lift G\"{o}del's soundness proof to the higher-type theory $\WEPAomega+\QFAC$ of \emph{weakly-extensional} Peano arithmetic with the quantifier-free axiom of choice (see \cite{Kohlenbach(2008.0)} for details). On the other hand, for the addition of computationally non-trivial choice principles, such as the axiom of countable choice $\AC_{\NN}$, for arbitrary formulas, the primitive recursive functionals no longer suffice for soundness of the interpretation. In fact, over $\WEPAomega$ countable choice is strong enough to derive the full comprehension schema
\begin{equation*}\CA \ \colon \ \exists f^{\NN \to X}\forall n(f(n)=0\leftrightarrow A(n))\end{equation*}
and so the theory $\WEPAomega + \QFAC + \AC_{\NN}$ is already capable of formalising a large portion of mathematical analysis, and is thus considerably stronger than Peano arithmetic. Nevertheless, just a few years after G\"{o}del's paper, Spector \cite{Spector(1962.0)} proved that one could indeed extend the Dialectica interpretation to full classical analysis provided we add bar recursion to system $\systemT$.

\subsection{The countable choice problem}
\label{sec-dialectica-problem}

Spector's main idea can be appreciated from a completely abstract perspective, independently of the full details of the Dialectica interpretation. Spector observed that in order to extend the ND interpretation to $\WEPAomega + \QFAC + \AC_{\NN}$, it suffices to find some way of realizing the Dialectica interpretation of the double negation shift, a non-constructive principle given by
\begin{equation*}\DNS \ \colon \ \forall n\neg\neg B(n)\to\neg\neg\forall n B(n).\end{equation*}
Now, suppose that the Dialectica interpretation of $B(n)$ is $\dt{B(n)}{x}{y}$ where $x\colon X$ and $y\colon Y$ are tuples of variables of the appropriate type. Then the Dialectica interpretation of $\DNS$ is given by
\begin{equation*}\dt{\DNS}{f,p,n}{\varepsilon,q,\varphi}=\dt{B(n)}{\varepsilon_np}{p(\varepsilon_np)}\to\dt{B(\varphi f)}{f(\varphi f)}{qf}.\end{equation*}
In other words, to solve the Dialectica interpretation of $\DNS$, for each given formula $B$ one must produce realizers $f\colon X^\NN$, $p\colon X\to Y$ and $n\colon\NN$ in terms of the parameters $\varepsilon\colon\NN\to (X\to Y) \to X$, $q\colon X^\NN\to Y$ and $\varphi\colon X^\NN\to\NN$ satisfying $\dt{\DNS}{f,p,n}{\varepsilon,q,\varphi}$. Spector approached this by tackling a stronger problem, namely to solve the underlying system of equations 
\begin{equation}\label{eqn-spector}
\begin{aligned}
\varphi f&=n \\
f n&=\varepsilon_n p \\
qf &= p(\varepsilon_np)
\end{aligned}
\end{equation}
in $f$, $p$ and $n$. We call the equations (\ref{eqn-spector}) \emph{Spector's equations}, and the issue of solving them the \emph{countable choice problem}. It is clear that a solution Spector's equations is also a realizer for $\DNS$, even independent of the formula $B$, and thus to extend the ND interpretation to classical analysis it suffices to find a general solution to the countable choice problem.

\subsection{Spector's bar recursive solution}
\label{sec-dialectica-spector}

Spector's classic solution to the countable choice problem was to use the general bar recursion $\BR$ to construct a term $\SBR_X^{\varepsilon,q,\varphi}$ in the parameters $\varepsilon$, $q$ and $\varphi$ of the problem, which satisfies the defining equation
\begin{equation*}
\SBR_X^{\varepsilon,q,\varphi}(s) =_{X^*} s \at
\begin{cases}
\pair{} & \mbox{if $\varphi(\ext{s})<|s|$} \\
\SBR_X^{\varepsilon,q,\varphi}(s\ast a_s) & \mbox{otherwise}.
\end{cases}
\end{equation*}
where $a_s := \varepsilon_{|s|}(\lambda x . \ext{q}(\SBR_X^{\varepsilon,q,\varphi}(s\ast x)))$. Actually Spector's defines a slightly different (but equivalent) variant of $\SBR_X^{\varepsilon,q,\varphi}(s)$, but precise details are not particularly relevant here (see e.g. \cite{OliPow(2012.2)}).
\begin{theorem}[\cite{Spector(1962.0)}]\label{thm-spector}Define
\begin{equation*}
\begin{aligned}
t &:=_{X^\ast}\SBR_X^{\varepsilon,q,\varphi}(\pair{}) \\
p_i &:=_{X\to Y}\lambda x . \ext{q}(\SBR_X^{\varepsilon,q,\varphi}(\initSeg{t}{i}\ast x))
\end{aligned}
\end{equation*}
where $i<|t|$ in the second equation. Then for all $0\leq i<|t|$ we have
\begin{equation} \label{eqn-spector-SBR}
\begin{aligned}
t_i &=\varepsilon_i p_i \\
\ext{q} t &= p_i(\varepsilon_ip_i).
\end{aligned}
\end{equation}
\end{theorem}

\begin{proof}This is a standard induction argument see e.g. \cite{Kohlenbach(2008.0),Oliva(2006.2)}.\end{proof}

\begin{corollary} \label{cor-spector}
Define $t$ and $p_i$ in the parameters $\varepsilon$, $q$ and $\varphi$ as in Theorem \ref{thm-spector}. Then $f:=\ext{t}$, $n:=\varphi f$, and $p:=p_n$ solve Spector's equations (\ref{eqn-spector}).
\end{corollary}
\begin{proof} The proof essentially reduces to verifying that $\varphi(\ext{t})<|t|$. The result then follows directly from Theorem \ref{thm-spector} and the equations (\ref{eqn-spector-SBR}). A sketch of the argument is as follows. Note that $\Phi(s)$ is an extension of $s$ and hence $|\Phi(s)| \geq |s|$. One first shows that if $t=\Phi(\pair{})$ then $\Phi(\pair{})=\Phi(t)$. Hence, if $\varphi(\ext{t})\geq |t|$ we would have that $\Phi(t)=\Phi(t\ast a_t)$ and thus $|t|=|\Phi(\pair{})|=|\Phi(t)|=|\Phi(t * a_t)|\geq |t\ast a_t|=|t|+1$, a contradiction. \end{proof}

\subsection{A symmetric solution}
\label{sec-dialectica-symmetric}

We now present our alternative solution to the countable choice problem which is based on our symmetric bar recursor $\sBR$ instead of the usual Spector recursor $\BR$. Our first step is to define a symmetric version of the special recursor $\SBR_X^{\varepsilon,q,\varphi}(s)$, which takes parameters $\varepsilon$, $q$ and $\varphi$ of the same type as those in for $\SBR$, but whose input and output are now \emph{finite partial functions}:
\begin{equation*}
\sSBR_X^{\varepsilon,q,\varphi}(u^{X^\dagger}) =_{X^\dagger}u\at
\begin{cases}
\emptyset & \mbox{if $n_u \in \dom(u)$} \\
\sSBR_X^{\varepsilon,q,\varphi}(u\oplus(n_u,a_u)) & \mbox{otherwise}
\end{cases}
\end{equation*}
where $n_u=\varphi(\ext{u})$ and $a_u=\varepsilon_{n_u}(\lambda x . \ext{q}(\sSBR_X^{\varepsilon,q,\varphi}(u\oplus(n_u,x)))$. We note without proof that $\sSBR$ is indeed definable from $\sBR$:
\begin{proposition} The functional $\sBR^{\phi^{\varepsilon,q,\varphi},\lambda \alpha . \alpha,\varphi}_{X,X^\dagger}(u)$, where
\eqleft{\phi^{\varepsilon,q,\varphi}_u(p^{X\to X^\dagger}) :=_{X^\dagger}u\at p(\varepsilon_{\varphi(\ext{u})}(\lambda x . \ext{q}(p(x)))),}
satisfies the defining equation of $\sSBR^{\varepsilon,q,\varphi}_X(u)$, provably in $\EHAomega$.\end{proposition}

Our construction and verification of a solution to Spector's equations now broadly follows, but is somewhat more intricate, than the usual approach for Spector's standard bar recursion. Recall the notion of a $\varphi$-thread from Definition \ref{defn-Sprop}.

\begin{lemma} \label{lem-symmetric} Assume $u$ is a $\varphi$-thread\footnote{Here the restriction $S_\varphi(u)$ on $u$ is merely a convenience as opposed to a necessity, allowing us to smoothly import the notation from Definition \ref{defn-S}.}. Let $v := \sSBR^{\varepsilon,q,\varphi}_X(u)$.Then
\begin{equation}\label{eqn-lem-symmetric}v=\sSBR_X^{\varepsilon,q,\varphi}(\initSegs{\varphi}{v}{i})\end{equation}
for all $|\dom(u)|\leq i\leq |\dom(v)|$.\end{lemma}
\begin{proof} By induction on $i$. \\
For $i=|\dom(u)|$ we claim that $\initSegs{\varphi}{v}{i}=u$. By a separate easy induction on $j \leq |\dom(u)|$, one can show that whenever $u\sqsubseteq v$ and $S_{\varphi}(u)$ then $\initSegs{\varphi}{u}{j}=\initSegs{\varphi}{v}{j}$ for all $0\leq j\leq |\dom(u)|$. Now in this case $u\sqsubseteq v$ by definition, and thus setting $j=|\dom(u)|$ we have $\initSegs{\varphi}{v}{|\dom(u)|}=u$. \\
Now, for the main induction step, assume that (\ref{eqn-lem-symmetric}) is true for some $|\dom(u)|\leq i<|\dom(v)|$. Because $i$ is strictly smaller than $|\dom(v)|$ it must be the case that $n_{\initSegs{\varphi}{v}{i}} := \ext{\varphi}(\initSegs{\varphi}{v}{i})\notin\dom(\initSegs{\varphi}{v}{i})$. 
%
Therefore, by the defining equation of $\Psi$ we obtain
\begin{equation*}
v \stackrel{\textup{\tiny{I.H.}}}{=}\sSBR(\initSegs{\varphi}{v}{i})=\initSegs{\varphi}{v}{i}\at \Psi(\initSegs{\varphi}{v}{i}\oplus (n_{\initSegs{\varphi}{v}{i}}, a_{\initSegs{\varphi}{v}{i}}))=\Psi(\initSegs{\varphi}{v}{i}\oplus (n_{\initSegs{\varphi}{v}{i}},a_{\initSegs{\varphi}{v}{i}})).
\end{equation*}
The last equality above holds because by definition $\initSegs{\varphi}{v}{i} \sqsubseteq \sSBR(\initSegs{\varphi}{v}{i}\oplus (n_i,a_{\initSegs{\varphi}{v}{i}}))$ and thus $\initSegs{\varphi}{v}{i}$ can be absorbed by the latter term. Now, observing that
$$a_{\initSegs{\varphi}{v}{i}}=\sSBR(\initSegs{\varphi}{v}{i}\oplus (n_{\initSegs{\varphi}{v}{i}}, a_{\initSegs{\varphi}{v}{i}}))(n_{\initSegs{\varphi}{v}{i}}) = v(n_{\initSegs{\varphi}{v}{i}})$$
we have  $v= \sSBR(\initSegs{\varphi}{v}{i}\oplus (n_{\initSegs{\varphi}{v}{i}},v(n_{\initSegs{\varphi}{v}{i}})))=\sSBR(\initSegs{\varphi}{v}{i+1})$, which completes the induction step.
\end{proof}

\begin{theorem}\label{thm-symmetric} Define
\begin{equation*}
\begin{aligned}
v&:=_{X^\dagger}\sSBR_X^{\varepsilon,q,\varphi}(\emptyset) \\
p_i&:=_{X\to Y}\lambda x . \ext{q}(\sSBR_X^{\varepsilon,q,\varphi}(\initSegs{\varphi}{v}{i}\oplus (n_i,x)))
\end{aligned}
\end{equation*}
where in the second equation $i<|\dom(v)|$ and $n_i=\ext{\varphi}(\initSegs{}{v}{i})$. Then for all $0\leq i< |\dom(v)|$ we have
\begin{equation}\label{eqn-spector-sSBR}\begin{aligned}v(n_i) &=\varepsilon_{n_i} p_i \\ \ext{q}(v) &= p_i(\varepsilon_{n_i}p_i).\end{aligned}\end{equation}\end{theorem}

\begin{proof} By Lemma \ref{lem-symmetric} we have that $v=\Psi(\initSegs{\varphi}{v}{i})$ for all $0\leq i\leq |\dom(v)|$. Using this fact, we can show, analogously to the proof of Lemma \ref{lem-symmetric}, that $n_i = \ext{\varphi}(\initSegs{\varphi}{v}{i})\notin\dom(\initSegs{\varphi}{v}{i})$ for any $0\leq i<|\dom(v)|$. Therefore
\begin{equation*}
v(n_i)=\Psi(\initSegs{\varphi}{v}{i})(n_i) = (\initSegs{\varphi}{v}{i}\at\Psi(\initSegs{\varphi}{v}{i}\oplus (n_i,a_{\initSegs{\varphi}{v}{i}})))(n_i)=a_{\initSegs{\varphi}{v}{i}}.
\end{equation*}
But by the definitions of $a_u$ and $p_i$ we have
\begin{equation*}
a_{\initSegs{\varphi}{v}{i}}=\varepsilon_{n_i}(\lambda x . \ext{q}(\Psi(\initSegs{\varphi}{v}{i}\oplus (n_i,x))))=\varepsilon_{n_i} p_i.
\end{equation*}
Put together these establish the first equation of (\ref{eqn-spector-sSBR}). For the second we have, once more using Lemma \ref{lem-symmetric}, that for any $0\leq i<|\dom(v)|$:
\begin{equation*}\ext{q}(v)=\ext{q}(\Psi(\initSegs{\varphi}{v}{i+1}))=\ext{q}(\Psi(\initSegs{\varphi}{v}{i}\oplus (n_i,v(n_i)))=p_i(v(n_i)),\end{equation*}
and thus by the first equation we have $\ext{q}(v)=p_i(\varepsilon_{n_i}p_i)$.
\end{proof}

\begin{corollary}\label{cor-symmetric} Define $v$ and $p_i$ in the parameters $\varepsilon$, $q$ and $\varphi$ as in Theorem \ref{thm-symmetric}. Let $k < |\dom(v)|$ be such that $n_k = \ext{\varphi}(v)$. Then $f:=\ext{v}$, $n:=n_k$, and $p:=p_{k}$ solve Spector's equations (\ref{eqn-spector}).
\end{corollary}

\begin{proof}First we must show that the index $k$ is well-defined. By an easy induction similar to all those in the preceding proofs, we can show that for $0\leq i\leq |\dom(v)|$ we have
\begin{equation*}\initSegs{\varphi}{v}{i}=(n_0,v(n_0))\oplus\ldots\oplus(n_{i-1},v(n_{i-1}))\end{equation*}
for distinct $n_j$, where as always $n_j=\ext{\varphi}(\initSegs{}{v}{j})$. In particular we have $S_\varphi(v)$ and thus $v=\initSegs{}{v}{|\dom(v)|}$. Now, since $\ext\varphi(v)\notin\dom(v)$ would imply that $|\dom(v)|=|\dom(\Psi(\initSegs{}{v}{|\dom(v)|}))|>|\dom(v)|$ we must have $\ext{\varphi}(v)\in\dom(v)$ i.e. $\ext{\varphi}(v)=n_k$ for some $k<|\dom(v)|$. The solution is now easily verified: $n=n_k=\ext{\varphi}(v)=\varphi(f)$ by definition, and by the equations (\ref{eqn-spector-sSBR}) we have $f(n) = v(n_k) = \varepsilon_{n_k}(p_k) = \varepsilon_n(p)$ and $qf=\ext{q}(v)=p_k(\varepsilon_{n_k}p_k)=p(\varepsilon_np)$.\end{proof}

To summarise, so far in this section we have outlined Spector's well-known reduction of the problem of realizing the extension of the ND interpretation of classical analysis to that of solving the set of equations (\ref{eqn-spector}). We then recounted his standard bar recursive solution to these equations, and followed this with a novel solution using a new, symmetric variant of bar recursion. So what is the essential difference between these two approaches?

Spector's solution to the countable choice works by computing finite sequences $s$ such that $f=\ext{s}$ forms a solution to the last pair of equations in (\ref{eqn-spector}) for all $n<|s|$, terminating once we have such a sequence which in addition satisfies $\varphi(\ext{s})<|s|$, thus allowing us to incorporate the first equation. This method eventually succeeds as long as we are working in a model in which $\SPEC$ holds, ensuring well-foundedness of the underlying computation tree. However, while the solution given by bar recursion is elegant in its simplicity, from a algorithmic perspective it is potentially inefficient, precisely because solutions are always computed for the last two of Spector's equations for all $n<|s|$ whereas we only need these equations to hold for $n = \varphi(\ext{s})$.

Our new method of constructing solutions to Spector's equations uses a new algorithm which constructs finite partial functions $u$ such that $f=\ext{u}$ forms a solution to the last pair of equations for all $n\in\dom(u)$, where the set $\dom(u)$ is guided by the parameter $\varphi$. This means that, in stark constrast to Spector's method, we do not necessarily need to have computed solutions for the last equations for every $n$ in some initial segment of $\NN$, but only for certain values. While our solution is somewhat more complicated to verify, and in particular is based on a form of recursion for which it is seemingly more difficult to prove termination (cf. Section \ref{sec-bar}), from a purely practical perspective it is possible that it gives rise to much more efficient programs.

We now present a short and informal case study in order to provide a concrete illustration of the differences between the two methods of program extraction. Our aim is to highlight that in practice the realizers based on symmetric bar recursion compare favourably to the traditional Spector bar recursion.

\section{Case study: No injection from $\NN\to\NN$ to $\NN$}
\label{sec-noinjection}
 
We illustrate the preceding theoretical results with a program extraction from the proof that there is no injection from the function space $\NN\to\NN$ to the natural numbers $\NN$. This theorem can be formalized as a $\Pi_2$-statement in the language of $\EPAomega$, and moreover its standard proof by a diagonal argument can be formalized using an instance of $\AC_{\NN,\NN\to\NN}$. 

This example was originally used by the first author in \cite{Oliva(2006.2)} in order to demonstrate the extraction of programs from proofs using Spector's bar recursor, and is a good candidate for a case study as it is relatively straightforward without being completely trivial. 
\begin{theorem}\label{thm-noinjection}$\EPAomega+\AC_{\NN,\NN\to\NN}\vdash\forall H\colon \NN^\NN\to\NN\; \exists \alpha,\beta\colon\NN\to\NN\; \exists i\; \colon\NN(\alpha i\neq \beta i\wedge H\alpha=H\beta)$.\end{theorem}
\begin{proof}[Classical Proof]As a simple case of the law of excluded middle (also known as the drinker's paradox) we have
\begin{equation}\label{eqn-noinjection-a}\forall n^\NN\exists \alpha^{\NN\to\NN}(\exists \beta(H\beta=n)\to H\alpha=n).\end{equation}
Applying $\AC_{\NN,\NN\to\NN}$ to (\ref{eqn-noinjection-a}) yields a functional $f\colon\NN\to\NN^\NN$ satisfying
\begin{equation}\label{eqn-noinjection-b}\forall n(\exists \beta(H\beta=n)\to H(f(n))=n).\end{equation}
The map $f$ produces for each $n$ a function $f(n)\colon\NN\to\NN$ such that whenever $n$ is in the range of $H$, $f(n)$ maps to $n$. Now, define $\alpha_H := \lambda n.f(n)(n)+1$ and let $i_H := H(\alpha_H)$ and $\beta_H := f(i_H)$. Then since $i_H$ is in the range of $H$, by (\ref{eqn-noinjection-b}) we must have $H(\beta_H) = H(f(i_H)) = i_H = H(\alpha_H)$. But $\alpha_H\neq f(i_H) = \beta_H$. \end{proof}

It is an intriguing consequence of Spector's ND interpretation of classical analysis that we are \emph{a priori} guaranteed to be able to convert the classical diagonal argument above into a semi-intuitionistic proof, and directly construct (using bar recursion) explicit witnesses for $\alpha_H$, $\beta_H$ and $i_H$ as a function of $H$. Moreover, Spector's reduction of the ND interpretation of analysis to the countable choice problem demonstrates that a realizer for Theorem \ref{thm-noinjection} can be constructed primitive recursively in an \emph{arbitrary} solution to the equations (\ref{eqn-spector}). In particular, we can replace Spector's bar recursion with an instance of symmetric bar recursion to give an alternative procedure for refuting injectiveness.

\begin{proposition}\label{prop-noinjection-ND}Any computable solution to Spector's equations allows us to effectively extract witnesses for $f$, $g$ and $i$ in Theorem \ref{thm-noinjection}.\end{proposition}

\begin{proof} 
A solution to Spector's equations acts as a computational interpretation of an instance of $\AC_{\NN}$. We must simply produce suitable parameters for these equations which correspond to the particular instance of $\AC_{\NN,\NN\to\NN}$ used in the classical proof. First, we give a computation interpretation to the initial instance of law of excluded middle via the term $\varepsilon\colon\NN\to (\NN^\NN\to\NN^\NN)\to\NN^\NN$ given by
\begin{equation}\label{eqn-noinjection-ND-a}
\varepsilon_n(p^{\NN^\NN\to\NN^\NN}) :=_{\NN^\NN}
\begin{cases}
\zero_{\NN\to\NN} & \mbox{if $H(p(\zero))\neq n$} \\[1mm] 
p(\zero) & \mbox{if $H(p(\zero))=n$}.
\end{cases}
\end{equation}
It is easy to verify that $\varepsilon$ satisfies 
\begin{equation}\label{eqn-noinjection-c}\forall n,p(H(p(\varepsilon_np))=n\to H(\varepsilon_np)=n),\end{equation}
which is just the ND interpretation of (\ref{eqn-noinjection-a}). Now, a computable solution to Spector's equations allows us to effectively construct an approximation to a choice sequence $f$ in the variables $q\colon (\NN\to\NN^\NN)\to\NN^\NN$ and $\varphi\colon (\NN\to\NN^\NN)\to\NN$ that satisfies
\begin{equation} \label{eqn-noinjection-d}
H(q(f))=\varphi(f)\to H(f(\varphi f))=\varphi(f)
\end{equation}
using (\ref{eqn-noinjection-c}) and Spector's equations (\ref{eqn-spector}). Therefore, defining $q(f)=\lambda n.f(n)(n)+1$ and $\varphi(f)=H(q(f))$, the premise of (\ref{eqn-noinjection-d}) holds by definition and hence we obtain $H(f(\varphi f))=\varphi(f)$. Finally, setting $\alpha_H:=q(f)$ and $\beta_H:=f(\varphi f)$ we have $H\beta_H=H\alpha_H$, but $\alpha_H$ and $\beta_H$ differ at $i_H:=\varphi(f)$. \end{proof}  

In a general manner of speaking, the reason one is are able to convert the classical proof of Theorem \ref{thm-noinjection} into a construction which computes $\alpha$ and $\beta$ for any given $H$ is that the solution of Spector's equations will typically only work in a subset of the full set-theoretic type structure. Solutions can be obtained, for instance, if either continuity or majorizability is assumed (cf. \cite{Bezem(1985.0),Scarpellini(1971.0)}), although Proposition \ref{prop-noinjection-ND} provides a solution that is independent of the particular model one has in mind. 

However, the exact computational process in calculating these witnesses will depend on our chosen solution of Spector's equations. We now briefly analyse the program which arises from choosing our symmetric bar recursive solution in place of Spector's original bar recursion.

\subsection{Numerical performance on sample input}
\label{sec-noinjection-numerical}

We have implemented both Spector's and our new variant of bar recursion in Haskell\footnote{\url{http://www.eecs.qmul.ac.uk/~pbo/code/bar/}}. We ran several tests in which we used both variants to compute counterexample functions $\alpha_H$ and $\beta_H$ as described in Proposition \ref{prop-noinjection-ND} for various concrete choices of $H$. To make things clearer, recall that
\begin{equation*}\alpha_H=q(f) \hspace{20mm} \beta_H=f(H(q(f))) \end{equation*}
where $f$ is some solution to Spector's equations (\ref{eqn-spector}) in $\varepsilon$, $q$ and $\varphi$, which are all defined as in the proof of the Proposition. By Theorems \ref{thm-spector} and \ref{thm-symmetric} our counterexamples for the two forms of bar recursion are given as follows:
\[
\begin{array}{rlll}
\mbox{Spector:} & \ \ \ \ \alpha_H = \ext{q}(\Phi^{\varepsilon,q,\varphi}(\pair{})) & \beta_H =\Phi^{\varepsilon,q,\varphi}(\pair{})(i_H) & \ \ \ \ i_H=H(\ext{q}(\Phi^{\varepsilon,q,\varphi}(\pair{}))) \\[2mm]
\mbox{Symmetric:} & \ \ \ \ \alpha_H = \ext{q}(\Psi^{\varepsilon,q,\varphi}(\emptyset)) & \beta_H =\Psi^{\varepsilon,q,\varphi}(\emptyset)(i_H) & \ \ \ \ i_H=H(\ext{q}(\Psi^{\varepsilon,q,\varphi}(\emptyset))).
\end{array}
\]
For each instance of $H$ we calculated 
\begin{enumerate}[(a)]

\item the domain size of the finite approximations $\Phi(\pair{}) \colon (\NN^\NN)^*$ and $\Psi(\emptyset) \colon (\NN^\NN)^\dagger$, and 

\item the number of recursive calls triggered when computing $i_H$ together with the values of $\alpha_H$ and $\beta_H$ up to this point.

\end{enumerate} 
Of course, we could have chosen other benchmarks by which to compare the realizers, although we regard this as being fairly representative of how we might want to use the realizer in practice. In any case our only, modest aim in this section is to give an informal comparison between the two methods of program extraction.

%
In the vast majority of natural cases we have considered the procedure based on symmetric bar recursion outperformed that based on Spector's bar recursion by a considerable margin. The interested reader is encouraged to try their own examples using our source code to witness this for themselves. However, we sketch a few representative examples here. \\

\noindent\textbf{Example 1.} Take the family of functionals $H_n$ defined by
%
\eqleft{H_n(\gamma^{\NN\to\NN})=\Pi^{n-1}_{i=0}(1+\gamma i).} 
The table below indicates the domain sizes of $\SBR(\pair{})$ and $\sSBR(\emptyset)$, and the number of recursive calls needed to calculate $i_H$ and $\initSeg{\alpha_H}{i_H+1}$ and $\initSeg{\beta_H}{i_H+1}$, for $n \in \{4,5,6\}$.
\begin{center}
\begin{tabular}{l|cc}
 & {\rm Spector} (domain size / \# recursive calls) & {\rm Symmetric} (domain size / \# recursive calls) \\
 \hline \\[-3mm]
 $n = 4$ & 17\; / \, \;1140 & 1\;/\;12 \\[1mm]
 $n = 5$  & 33\; / \;4650 & 1\;/\;12 \\[1mm]
 $n = 6$ & \, 65\; / \;19154 & 1\;/\;12
\end{tabular}
\end{center}
This disparity in performance can be intuitively understood by computing on paper what each recursor does. First, take the symmetric recursor. One can show that
\begin{equation*}\ext{\varphi}(\emptyset)=H_n(q(\zero_{\NN\to\NN\to\NN}))=H_n(\lambda i.1)=\Pi^{n-1}_{i=0}(1+1)=2^n\end{equation*} 
and thus $\Psi(\emptyset)=\Psi((2^n,a_{\emptyset}))$ where
\begin{equation*}a_\emptyset:=\varepsilon_{2^n}(\lambda x\; . \; \ext{q}(\Psi((2^n,x))))=\begin{cases}\zero_{\NN\to\NN} & \mbox{if $H_n(\ext{q}(\Psi((2^n,\zero))))\neq 2^n$}\\ \ext{q}(\Psi((2^n,\zero))) & \mbox{otherwise}.\end{cases}\end{equation*}
But since $\ext{\varphi}((2^n,\zero))=\ext{\varphi}(\emptyset)=2^n\in\dom((2^n,\zero))=\{2^n\}$ we see that $\Psi((2^n,\zero))=(2^n,\zero)$, and therefore $\ext{q}(\Psi((2^n,\zero)))=\lambda k.1$ and thus $a_{\emptyset}=\lambda k.1$. Now, it is easy to show that we also have
\begin{equation*}\ext{\varphi}((2^n,\lambda k.1))=\Pi^{n-1}_{i=0}(1+1)=2^n\end{equation*}
and thus the recursion terminates, yielding $\Psi(\emptyset)=\Psi((2^n,\lambda k.1))=(2^n,\lambda k.1)$. Our counterexample are then
\begin{equation*}\begin{aligned}\alpha_{H_n}&=\ext{q}((2^n,\lambda k.1))=\lambda k.\mbox{ if $k=2^n$ then $2$ else $1$}\\
\beta_{H_n}&=\Psi(\emptyset)(H_n(\alpha_{H_n}))=\Psi(\emptyset)(2^n)=\lambda k.1 \end{aligned}\end{equation*}
and it is easy to verify that these counterexamples indeed work. Moreover, in order to compute them our symmetric bar recursor only needed to produce a finite partial function with one element, resulting in a very quick algorithm.

In contrast, take Spector's bar recursor. Entirely analogously to before we have
\begin{equation*}\ext{\varphi}(\pair{})=H_n(q(\zero_{\NN\to\NN\to\NN}))=2^n\end{equation*}
but now since the recursor is forced to compute sequentially we have $\Phi(\pair{})=\Phi(\pair{a_{\pair{}}})$ where
\begin{equation*}a_{\pair{}}:=\varepsilon_0(\lambda x\; . \; \ext{q}(\Phi(\pair{x})))=\begin{cases}\zero_{\NN\to\NN} & \mbox{if $H_n(\ext{q}(\Phi(\pair{\zero})))\neq 0$} \\ \ext{q}(\Phi(\pair{\zero}))\neq 0 & \mbox{otherwise}. \end{cases}\end{equation*}
Since $\ext\varphi{\pair{\zero}}=\ext{\varphi}(\pair{})=2^n>|\pair{\zero}|$ we have that $\Phi(\pair{\zero})=\pair{\zero}\at\Phi(\pair{\zero,a_{\pair{\zero}}})$ and we must continue by computing $a_{\pair{\zero}}$. Expanding the definition of $a_{\pair{\zero}}$ in terms of $\varepsilon_1$ analogously requires a further recursive call to $\Phi(\pair{\zero,\zero})$, and since once more $\ext{\varphi}{\pair{\zero,\zero}}=\ext{\varphi}(\pair{})=2^n>|\pair{\zero,\zero}|$ we have $\Phi(\pair{\zero,\zero})=\pair{\zero,\zero}\at\Phi(\pair{\zero,\zero,a_{\pair{\zero,\zero}}})$ and so on. In fact, our recursor continues to make nested recursive calls in this fashion until it reaches the list $\pair{\zero,\ldots,\zero}$ of size $2^{n}$, in which case one can compute that $\Phi(\pair{\zero\ldots\zero})=\pair{\zero,\ldots,\zero,\lambda k.1}$. Backtracking one finally obtains 
\begin{equation*}\Phi(\pair{})=\Phi(\pair{\zero})=\ldots=\pair{\underbrace{\zero,\ldots,\zero}_{\mbox{$2^n$ times}},\lambda k.1}.\end{equation*}
Our counterexamples are then
\begin{equation*}\begin{aligned}\alpha_{H_n}&=\ext{q}(\pair{\zero,\ldots,\zero,\lambda k.1})=\lambda k.\mbox{ if $k=2^n$ then $2$ else $1$}\\
\beta_{H_n}&=\Psi(\emptyset)(H_n(\alpha_{H_n}))=\Psi(\emptyset)(2^n)=\lambda k.1 \end{aligned}\end{equation*}
which are exactly the same as those produces by the symmetric bar recursor. However, the algorithm induced by the Spector recursor was forced to carry out a lengthy backtracking procedure along sequences of length up to $2^{n+1}$, resulting in a much more complex computation. 

%
%

If we adjust $H_n$ to make it more complex still, for example 
\eqleft{H_n(\gamma)=\Pi^{n-1}_{i=0}(1+i)^{1+\gamma i},}
the disparity is even more extreme:
\begin{center}
\begin{tabular}{l|cc}
 & {\rm Spector} & {\rm Symmetric} \\
 \hline \\[-3mm]
 $n = 3$ & 577\;/2350 & 1\;/\;12 \\[1mm]
 $n = 4$ & 577\;/365700 & 1\;/\;12
\end{tabular}
\end{center}

\noindent\textbf{Example 2.} Now suppose $H_n(\gamma)$ searches for a least point such that $\gamma i<\gamma(i+1)$:
\begin{equation*}H_n(\gamma)=\mbox{least $i\leq n$ such that $\gamma i<\gamma(i+1)$, else $n$ if none exist}.\end{equation*} 
The corresponding data for $n\in\{3,4,5\}$ is now:
\begin{center}
\begin{tabular}{l|cc}
& {\rm Spector} & {\rm Symmetric} \\
\hline \\[-3mm]
$n = 3$ & \, 4\;/\;316 & 4\;/\;52 \\[1mm]
$n = 4$ & \, 5\;/\;688 & 5\;/\;64 \\[1mm]
$n = 5$ & \;\; 6\;/\;1444 & 6\;/\;76
\end{tabular}
\end{center}
In this case, both recursors terminate once they have computed a domain of size $n$, but Spector's recursor takes much longer. Once more, this behaviour has an intuitive explanation. Let us take $n=3$ for simplicity, and define functions $\NN\to\NN$ via the following notational convention, whereby
\begin{equation*}\gamma:=[x_0,x_1,\ldots,x_{k-1}]\end{equation*} 
means that $\gamma i=x_i$ for $i< k$ and $\gamma i=1$ otherwise. Let's first look at the symmetric recursor. We have 
\begin{equation*}\ext{\varphi}(\emptyset)=H_3(q(\zero))=H_3([1,1,1,1])=3\end{equation*}
since $H_3$ is forced to return the `else' clause, and therefore $\Psi(\emptyset)=\Psi((3,\gamma_0))$ where we write $\gamma_0:=a_\emptyset=\varepsilon_3(\lambda x\; . \; \ext{q}(\Psi((3,x))))$. The functional $\varepsilon_3$ be expanded just as in the previous section, and analogously to there we can check that $\ext{q}(\Psi((3,\zero)))=\ext{q}((3,\zero))=[1,1,1,1]$ and thus $\gamma_0=[1,1,1,1]$. Now we have
\begin{equation*}\ext{\varphi}((3,\gamma_0))=H_3(\ext{q}((3,\gamma_0)))=H_3([1,1,1,2])=2.\end{equation*}
Thus $\Psi((3,\gamma_0))=\Psi((3,\gamma_0)\oplus (2,\gamma_1))$ where $\gamma_1:=a_{(3,f_0)}=\varepsilon_2(\lambda x\; . \; \ext{q}(\Psi((3,\gamma_0)\oplus (2,x))))$. Now again we can verify that $\ext{q}(\Psi((3,\gamma_0)\oplus (2,\zero)))=\ext{q}((3,\gamma_0)\oplus (2,\zero))=[1,1,1,2]$ and thus $f_1=[1,1,1,2]$. This process continues until we obtain
\begin{equation*}\Psi(\emptyset)=(0,[1,2,2,2])\oplus (1,[1,1,2,2])\oplus(2,[1,1,1,2])\oplus (3,[1,1,1,1]),\end{equation*}
from which we read off $\alpha_{H_3}=\ext{q}(\Psi(\emptyset))=[2,2,2,2]$ and $\beta_{H_3}=\Psi(\emptyset)(3)=[1,1,1,1]$, and it is clear that these are valid counterexamples.

Entirely analogously, for general $n$ we obtain 
\begin{equation*}\alpha_{H_n}=[\underbrace{2,2,\ldots,2}_{\mbox{\scriptsize $n+1$ times}}]\mbox{ \ \ \ and \ \ \ } \beta_{H_n}=[\underbrace{1,1,\ldots,1}_{\mbox{\scriptsize $n+1$ times}}]\end{equation*}
In the case of Spector's bar recursion, it is clear that since we always have $\varphi(f)=H_n(q(f))\leq n$ then the recursor will output a finite sequence of length at most $n$, and so the domain of the bar-recursive output is no bigger that that of symmetric bar recursion. However, as in the first example the computation itself involves a much more intricate backtracking procedure, and this time also yields a slightly different solution. By a somewhat tedious calculation (which the reader can try to reproduce themselves they want convinving of the complexity of Spector's recursion!) one obtains for e.g. $n=3$
\begin{equation*}\Phi(\pair{})=\pair{[1,2,1,1],[2,1,2,1],[2,2,1,2],[2,2,2,1]}\end{equation*}
yielding solutions $\alpha_{H_n}=[2,2,2,2]$ and $\beta_{H_n}=[2,2,2,1]$ and analogously for general $n$:
\begin{equation*}\alpha_{H_n}=[\underbrace{2,2,\ldots,2}_{\mbox{\scriptsize $n+1$ times}}]\mbox{ \ \ \ and \ \ \ } \beta_{H_n}=[\underbrace{2,2,\ldots,2}_{\mbox{\scriptsize $n$ times}},1]\end{equation*}
However, due again to its insistence of making recursive calls sequentially, Spector's recursor takes much longer to compute these solutions, as is clear from the table.\\

\noindent\textbf{Summary.} Naturally, this section remains very informal given that we have only provided a couple of examples to illustrate the difference between programs obtained using our recursor and those obtained via the traditional Spector bar recursor. Moreover, it is not the case that the symmetric recursor \emph{always} produces a more efficient algorithm than the Spector -- it is not too hard to come up with a somewhat contrived $H_n$ for which a sequential bar recursion is clearly better. Take, for example
\[
H_n(\gamma)=\begin{cases}
\mbox{greatest $i\leq n (\gamma(i)=1)$ if it exists, else $n$} & \mbox{if $\gamma(0)=\gamma(1)=2$}\\
0 & \mbox{if $\gamma(0)=1\wedge \gamma(1)=2$ or $\gamma(0)=2\wedge \gamma(1)=1$}\\
1 & \mbox{otherwise}.\end{cases}
\]
Here the sequential computation associated with $\SBR$ means that the first clause in the case distinction is never triggered, so that $\SBR$ always returns a sequence of length $2$. On the other hand, $\sSBR$ ends up with a finite partial function of size $n$, and so its cost is proportional to $n$. 

Nevertheless, the fact that such a $H$ exists does not necessarily detract from our symmetric bar recursion being a useful alternative to Spector's bar recursion in practice. In particular, when \emph{using} our realizer as, for instance, a building block for a more complex realizing term arising from a classical proof that uses Theorem \ref{thm-noinjection} as a lemma, it is reasonable to assume that $H$ will take the form of a fairly natural recursive function, and the authors conjecture that in many such cases symmetric bar recursion will drastically outperform Spector's bar recursion.

On top of this, due to the fact that it often avoids unnecessary backtracking, we believe that our symmetric recursor will often give rise to programs that are more natural and easier to understand from an \emph{algorithmic} perspective. In particular, we conjecture that there are close links between the symmetric bar recursive interpretation of choice over sequences law of excluded middle for $\Sigma_1$-formulas (of which the instance of $\AC_{\NN}$ here is an example) and the learning procedures for $\PA+{\sf EM}_1$ described in \cite{AscBer(2010.0)}.

However, we leave any further analysis to future work. For now, our main achievement has been to devise an interesting symmetric alternative to Spector's bar recursion, which appears to us more sensible and natural for the purpose of program extraction, and in a small number of test cases drastically outperforms the latter.

\section{Equivalence of $\BR$ and $\sBR$}
\label{sec-equiv}

In this section we prove that the recursion schemata $\BR$ and $\sBR$ are actually primitive recursively equivalent. This is the most technical part of the paper, but is entirely self-contained and is not at all necessary in order to understand the preceding sections. The most difficult direction -- the definability of $\sBR$ from $\BR$ -- can be carried out in $\EHAomega+\sBI$ and hence (by Theorem \ref{thm-sBI}) in $\EPAomega+\DC$, and thus as an immediate consequence we prove that $\sBI$ exists in any model of $\EPAomega+\DC$ which also validates $\BR$. In particular, $\sBI$ exists in both the Kleene/Kreisel total continuous functional (as stated earlier in Theorem \ref{thm-sBR-CONT}) but also the strongly majorizable functionals. 


\subsection{$\BR$ is definable from $\sBR$}

\begin{theorem}\label{thm-sBR-BR} $\BR_{X,R}$ is primitive recursively definable from $\sBR_{X \times \BB, R}$, provably in $\EHAomega + \sBR$.\end{theorem}

\begin{proof} Suppose we are given parameters $\phi\colon X^\ast\to (X\to R)\to R$, $b\colon X^\ast\to R$ and $\varphi\colon X^\NN\to \NN$ for $\BR_{X,R}$. Then there is a term $\Phi^{\phi,b,\varphi}$ primitive recursive in $\sBR_{X\times\BB,R}$ that satisfies the defining equation of $\BR_{X,R}^{\phi,b,\varphi}$.
Recall that we identify $\BB$ with the set $\{0,1\}$, taking $\zero_\BB:=0$, and thus $\zero_{X\times \BB}\equiv\pair{\zero_X,0}$. Define the map $\eta\colon X^\ast\to (X\times\BB)^\dagger$ by
\begin{equation*}
(\eta s)(n) := 
\begin{cases}
\pair{s_n,1} & \mbox{if $n<|s|$} \\
\mbox{undefined} & \mbox{otherwise},
\end{cases}
\end{equation*}
so that $\dom(\eta s) = \{ 0, 1, \ldots, |s| - 1 \}$; and conversely the map $\eta'\colon (X\times\BB)^\dagger \to X^\ast$ by $|\eta'u|=N+1$ where $N$ is the maximum element of $\dom(u)$, and
\begin{equation*}(\eta'u)_i:=\begin{cases}\pi_0(u(i)) & \mbox{if $i\in \dom(u)$}\\ \zero_X & \mbox{otherwise}\end{cases}\end{equation*}
where $\pi_0\colon X\times\BB\to X$ is the first projection. Note that $\eta'\eta s=s$ for all $s\colon X^\ast$. Now, define parameters $\tilde\phi$, $\tilde b$ and $\tilde\varphi$ for $\sBR_{X\times\BB,R}$ by
\begin{equation*}
\begin{aligned}
\tilde\varphi(\alpha^{(X\times\BB)^\NN})&\stackrel{\NN}{:=}\mu i\leq\varphi(\pi_0 \alpha)((\pi_1\alpha)(i)=_\BB 0) \\
\tilde b(u^{(X\times \BB)^\dagger})&\stackrel{R}{:=}b(\eta'u)\\
\tilde\phi_u(p^{X\times\BB\to R})&\stackrel{R}{:=}\phi_{\initSeg{\pi_0 \ext{u}}{\tilde\varphi(\ext{u})}}(\lambda x^X . p(\pair{x,1}))
\end{aligned}\end{equation*}
where $\pi_i\alpha$ denotes the projection of the sequence $\alpha$ i.e. $(\pi_i\alpha)(n):=\pi_i(\alpha(n))$. 
We claim that
\eqleft{\Delta^{\phi,b,\varphi}(s) := \sBR_{X \times \BB, R}^{\tilde\phi,\tilde b,\tilde\varphi}(\eta s)}
satisfies the defining equation of $\BR_{X, R}^{\phi,b,\varphi}(s)$. To prove this, first note that
\begin{itemize}
	\item[($i$)] $\pi_0 (\wext{\eta s}) =_{X^\NN}\ext{s}$, and
	\item[$(ii)$] $\pi_1 (\wext{\eta s})(i)=_\BB 1$, for $i<|s|$ and $0$ otherwise,
\end{itemize}
follow directly from the definition of $\eta$ and the fact that $\zero_{X\times\BB}=\pair{\zero_X,0}$. Therefore
\begin{itemize}
\item[($iii$)] $\tilde\varphi(\wext{\eta s})
	= \mu i\leq\varphi(\pi_0(\wext{\eta s})))(\pi_1(\wext{\eta s})(i)=_\BB 0)
	\stackrel{(i)}{=} \mu i\leq\varphi(\ext{s})(\pi_1(\wext{\eta s})(i)=0)
	\stackrel{(ii)}{=}
\begin{cases}
	\varphi(\ext{s}) & \mbox{if $\varphi(\ext{s})<|s|$} \\[1mm]
	|s| & \mbox{if $\varphi(\ext{s})\geq |s|$}.
\end{cases}$
\end{itemize}
Since $\dom(\eta s)=\{0,\ldots,|s|-1\}$ point ($iii$) above implies the equivalence
\begin{itemize}
	\item[($iv$)] $\varphi(\ext{s})<|s| \;\;\Leftrightarrow\;\; \tilde\varphi(\wext{\eta s})\in\dom(\eta s)$.
\end{itemize}
Therefore, if $\varphi(\ext{s})<|s|$ then $\tilde\varphi(\wext{\eta s})\in\dom(\eta s)$ and hence
\eqleft{
\;\;\Delta(s) = \sBR_{X \times \BB, R}^{\tilde\phi,\tilde b,\tilde\varphi}(\eta s) = \tilde b(\eta s)=b(\eta'\eta s)=b(s).
}
If $\varphi(\ext{s})\geq |s|$ then $\tilde\varphi(\wext{\eta s}) = |s| \notin\dom(\eta s)$ and hence
\eqleft{
\begin{array}{ll}
\Delta(s)&= \sBR_{X \times \BB, R}^{\tilde\phi,\tilde b,\tilde\varphi}(\eta s) \\[2mm]
	&=\tilde\phi_{\eta s}(\lambda \pair{x,b}^{X\times\BB} . \, \sBR(\upd{\eta s}{|s|}{\pair{x,b}})) \\[2mm]
	&=\phi_{\initSeg{\pi_0(\wext{\eta s})}{\tilde\varphi(\wext{\eta s})}}(\lambda x\,.\, \sBR(\upd{\eta s}{|s|}{\pair{x,1}})) \\
	&\stackrel{(i)}{=}\phi_{\initSeg{\ext{s}}{|s|}}(\lambda x\,.\, \sBR(\upd{\eta s}{|s|}{\pair{x,1}})) \\[2mm]
	&=\phi_s(\lambda x\,.\, \sBR(\eta(s\ast x))) \\[2mm]
	&=\phi_s(\lambda x\,.\, \Delta(s\ast x))
\end{array}
}
where for the penultimate equality one easily verifies that $\upd{\eta s}{|s|}{\pair{x,1}}=\eta(s\ast x)$. \end{proof}

The basic idea behind the preceding proof is that $\BR_{X,R}$ can be defined from a single instance of $\sBR_{X \times \BB,R}$ of (essentially) the same type, in which the symmetric control functional $\tilde\varphi$ is designed to be `stubborn' and always search for the least undefined point to update, thereby simulating Spector's bar recursion. 

\subsection{$\sBR$ is definable from $\BR$}
\label{sec-sBR-BR}

A similar idea, however, does not seem to work in the opposite direction, since Spector's bar recursion is inherently less flexible than symmetric bar recursion. Instead, to define $\sBR_{X,R}$ we resort to an instance of Spector's bar recursion of a strictly higher type, and the resulting construction is somewhat more intricate.

\begin{theorem}\label{thm-BR-sBR} $\sBR$ is primitive recursively definable from $\BR$, provably in $\EHAomega+\sBI+\BR$.\end{theorem}

We break up the proof of Theorem \ref{thm-BR-sBR} into several steps. The basic idea behind our construction is as follows: a finite partial state $u$ in the computation of $\sBR$ is represented by a finite sequence of pairs $s_u$ in our instance of $\BR$. If $n\in\dom(u)$ then the first component of $s_u(n)$ contains the state $u' \sqsubseteq u$ that was present when point $n$ was updated, and if $n\notin\dom(u)$ then the second component of $s_u(n)$ contains a continuation that allows us to make bar recursive calls on updates of the state at any point in the future allowing us to simulate the behaviour of $\sBR$. First, we need some definitions. 

\begin{definition}\label{defn-BR-sBR} Suppose $u,v\colon X^\dagger$ and $x\colon X$. Then $\initSeg{u}{n}\ast x\at v\colon X^\dagger$ denotes the finite partial function given by
\begin{equation*}
(\initSeg{u}{n}\ast x\at v)(i) = 
\begin{cases}
u(i) & \mbox{if $i<n$} \\
x & \mbox{if $i=n$} \\
v(i) & \mbox{if $i>n$}.
\end{cases}
\end{equation*}
So $\dom(\initSeg{u}{n} \ast x\at v) = (\dom(u) \backslash \{n,\ldots\}) \cup \{n\} \cup (\dom(v) \backslash \{ 0,\ldots, n \})$. Next, let us introduce the type abbreviation $Y:\equiv X^\dagger\times (X^\dagger\to (X\to R))$, and define the `diagonal' functional $d\colon Y^\ast\to X^\dagger$ by
\begin{equation*}
d(s)(j):=
\begin{cases}
(\pi_0s_i)(j) & \mbox{for least $i\leq j$ such that $i<|s|$ and $j\in\dom(\pi_0s_i)$} \\
\mbox{undefined} & \mbox{otherwise}.
\end{cases}
\end{equation*} 
The function $d$ returns a particular kind of `union' of the finite partial functions $\pi_0s_i$, where in the event that the $\pi_0s_i$ are defined at $j$ for more than one index $i\leq j$, the value at the least index is chosen. Because $s$ is a finite sequence the resulting partial function $d(s)$ must also have finite domain. Similarly, define an infinitary diagonal function $d^\infty\colon Y^\NN\to X^\NN$ by
\begin{equation*}
d^\infty(\alpha)(j):=
\begin{cases}
(\pi_0\alpha_i)(j) & \mbox{for least $i\leq j$ such that $j\in\dom(\pi_0\alpha_i)$} \\[1mm]
\zero_X & \mbox{otherwise}.
\end{cases}
\end{equation*}
This function $d^\infty$ returns a similar union of the infinite sequence of partial functions $\pi_0\alpha_i$, returning a partial function with potentially infinite domain, and then embedding this partial function in $X^\NN$ by assigning the canonical value $\zero_X$ to undefined elements.
\end{definition}

The main step in our proof of Theorem \ref{thm-BR-sBR} will be to show that $\BR$ defines a slightly altered form of $\sBR$, which only accepts $\varphi$-threads as input. As we show in Lemma \ref{lem-Psi-sBR} this restriction is inessential -- however, it makes the verification slightly easier to work with and hence we adopt this variant for now.

%

\newcommand{\sBRVar}{\Theta}

\begin{lemma}\label{lem-BR-sBR} Let $\sBRVar$ be the following variant of $\sBR$,
\begin{equation}
\label{sbr-var-def}
\sBRVar^{\phi,b,\varphi}(u^{X^\dagger}):=
\begin{cases}
\zero_R & \mbox{if $\neg S_\varphi(u)$} \\[1mm]
b(u) & \mbox{if $\varphi(\ext{u})\in\dom(u)$} \\[1mm]
\phi_u(\lambda x . \sBRVar^{\phi,b,\varphi}(u\oplus (\varphi(\ext{u}),x))) & \mbox{otherwise}.
\end{cases}
\end{equation}
where the parameters have the same type as those for $\sBR$, namely $\phi\colon X^\dagger\to (X\to R)\to R$, $b\colon X^\dagger\to R$ and $\varphi\colon X^\NN\to \NN$. Then $\sBRVar$ is primitive recursively definable from Spector's bar recursion $\BR_{Y,R}$ for $Y:\equiv X^\dagger\times (X^\dagger\to (X\to R))$, provably in $\EHAomega+\BR$. \end{lemma}

\begin{proof} Suppose now that we are given parameters $\phi$, $b$ and $\varphi$ for $\sBRVar$. We now define parameters $\tilde\phi$, $\tilde b$ and $\tilde\varphi$ for $\BR_{Y,R}$ in terms of $\phi$, $b$ and $\varphi$. We begin with $\tilde\phi$, which is given by
\begin{equation*}
\tilde\phi_{s^{Y^\ast}}(p^{Y\to R})\stackrel{R}{:=}
\begin{cases}
p(\pair{d(s),\zero_{X^\dagger\to(X\to R)}}) & \mbox{if $|s|\in\dom(d(s))$} \\[1mm]
p(\pair{d(s),\lambda v^{X^\dagger}, x^X . p(\pair{\initSeg{d(s)}{|s|}\ast x\at v,\zero_{X^\dagger\to (X\to R)}})}) & \mbox{otherwise}.
\end{cases}
\end{equation*}
As will become clear below, $p$ plays the role of a continuation: If $d(s)$ is not defined at $|s|$, it initiates a nested recursive call to variants of $d(s)$ of the form $\initSeg{d(s)}{|s|}\ast x\at v$, which are identical to $d(s)$ for all arguments $i<|s|$ but now defined with some value $x$ at $|s|$. For the parameter $\tilde b$, define
\begin{equation*}
\tilde b(s^{Y^\ast})\stackrel{R}{:=}
\begin{cases}
b(d(s)) & \mbox{if $\ext{\varphi}(d(s))\in\dom(d(s))$} \\[1mm]
\phi_{d s}(\pi_1(\ext{s}_{\ext{\varphi}(d(s))})(d(s))) & \mbox{otherwise}.
\end{cases}
\end{equation*}
Note that this is well typed since $\ext{s}\colon Y^\NN$, and thus $\ext{s}_{\varphi(\wext{d(s)})}\colon Y$ and $\pi_1\ext{s}_{\varphi(\wext{d(s)})}\colon X^\dagger\to (X\to R)$, which implies that $\pi_1(\ext{s}_{\varphi(\wext{d(s)})})(d s) \colon X\to R$. For the final parameter, let 
\begin{equation*}
\tilde\varphi(\alpha^{Y^\NN})\stackrel{\NN}{:=}\varphi(d^\infty(\alpha)).
\end{equation*}
We now define a sequence of finite sequences $s_{u, i}\colon Y^\ast$ for $u\colon X^\dagger$, primitive recursively in $\BR_{Y,R}^{\tilde\phi,\tilde b,\tilde \varphi}$, as
\begin{equation*}
\begin{aligned}
s_{u,0}&:=\pair{} \\ 
s_{u,i+1}&:=
\begin{cases}
\initSeg{s_{u,i}}{n_i}\ast\pair{u_{i+1},\zero_{X^\dagger\to (X\to R)}} & \mbox{if $n_i<|s_{u,i}|$}\\
s_{u,i}\ast\pair{d(s_{u,i}),f^i_{|s_{u,i}|}}\ast\ldots\ast\pair{d(s_{u,i}),f^i_{n_i-1}}\ast\pair{u_{i+1},\zero} & \mbox{otherwise}
\end{cases}
\end{aligned}
\end{equation*}
where we use the abbreviations $u_i := \initSegs{\varphi}{u}{i}$ and $n_i := \ext{\varphi}(u_i)$, and the functions $f^i_n$ are defined using course-of-values recursion as
\begin{equation} \label{br-defs-sbr-eq}
f^i_n:=_{X^\dagger\to(X\to R)}
\begin{cases}\lambda w^{X^\dagger}, x^X . \BR_{Y,R}^{\tilde\phi,\tilde b,\tilde\varphi}(\initSeg{s_{u,i+1}}{n}\ast\pair{\initSeg{d(s_{u,i})}{n}\ast x\at w,\zero}) & \mbox{if $n\notin\dom(d(s_{u,i}))$}\\ \zero & \mbox{otherwise}.
\end{cases}
\end{equation}
While the definition of $s_{u,i+1}$ may seem circular, in that $f^i_n$ uses $s_{u,i+1}$, it is well-defined by course-of-values recursion along the length of $s_{u,i+1}$ because $f^i_n = \pi_1(s_{u,i+1}(n))$, and to define this we only require knowledge of $\initSeg{s_{u,i+1}}{n}$ i.e. $s_{u,i+1}(m)$ for $m<n$, and for the base case $f^i_{|s_{u,i}|}$ is defined in terms of $\initSeg{s_{u,i+1}}{|s_{u,i}|}=s_{u,i}$. \\
Now, let $s_u:=s_{u,|\dom(u)|}$, so that $s_u \colon Y^*$. It is easy to see that $|s_{u, i+1}| = n_i + 1$ for arbitrary $i$, which means that in particular $|s_u| = n_{|\dom(u)|-1} + 1$ whenever $|\dom(u)|>0$. We claim that $\sBRVar$ can be defined from $\BR$ as follows:
\begin{equation} \label{eqn-Phi} \ \ \
\sBRVar(u):=\begin{cases}\zero_R & \mbox{if $\neg S_\varphi(u)$} \\[1mm]
\BR_{Y,R}^{\tilde\phi,\tilde b,\tilde\varphi}(s_u) & \mbox{otherwise}
\end{cases}
\end{equation}
satisfies the equation (\ref{sbr-var-def}). The rest of the section contains the lemmas needed to verify this claim. \end{proof}
The claim that $\sBRVar$ as defined in (\ref{eqn-Phi}) satisfies equation (\ref{sbr-var-def}) clearly holds in the case that $\neg S_\varphi(u)$, so from now on we assume that $S_\varphi(u)$ is true. 


\begin{lemma} \label{lemma-B} $d(s_{u,i})=u_i$, for all $i \leq |\dom(u)|$.
\end{lemma}
Lemma \ref{lemma-B} is not a deep result at all, and can be intuitively seen by inspecting the definition of $s_{u,i}$. However, due to the syntactic complexity of the underlying definitions, the result is quite tedious to prove, so we relegate this to Appendix \ref{app-A}.

\begin{lemma} \label{lemma-C} $d(s_u)=u$.
\end{lemma}
\begin{proof} Recall that we are assuming $S_\varphi(u)$. We have $d(s_u) = d(s_{u, |\dom(u)|}) \stackrel{\textup{L}\ref{lemma-B}}{=} u_{|\dom(u)|} \stackrel{\textup{L}\ref{lem-S}}{=} u$.
\end{proof}

\begin{lemma} \label{lemma-D} $d^\infty(\wext{s_u})=\ext{u}$, and hence $\tilde\varphi(\wext{s_u}) = \varphi(\ext{u})$.
\end{lemma}
\begin{proof} Recall that we assume an encoding of the type $X^\dagger$ such that $\zero_{X^\dagger}\equiv\emptyset$, and hence $\zero_Y=\pair{\zero_{X^\dagger},\zero_{X^\dagger\to (X\to R)}}$. It is clear from this that $d^\infty(\ext{s})=\wext{d(s)}$ for arbitrary $s\colon Y^\ast$. Hence, by Lemma \ref{lemma-C} we have $d^\infty(\wext{s_u})=\wext{d(s_u)}=\wext{u}$ and hence $\tilde\varphi(\wext{s_u})=\varphi(d^\infty(\wext{s_u}))=\varphi(\wext{u})$, by the definition of $\tilde\varphi$.
\end{proof}

\begin{lemma}\label{lemma-Ep}Let $i\leq |\dom(u)|$. Then if $n\notin\dom(u_i)$ and $n<|s_{u,i}|$ then $\pi_1((s_{u,i})_n)=f^k_n$ where $k$ is the least number such that $\forall j \in \{k, \ldots, i-1\} (n_j>n)$.\end{lemma}

\begin{proof}Induction on $i$. If $i=0$ then the claim is trivial since $|s_{u,0}|=0$. Suppose that the lemma is true for $i<|\dom(u)|$. Then as in the proof of Lemma \ref{lemma-B} there are two main cases corresponding to $n_i<|s_{u,i}|$ or $n_i\geq |s_{u,i}|$. \\
In the first case, suppose that $n\notin\dom(u_{i+1})$ and $n<|s_{u,i+1}|=n_i+1$. Then since $n\neq n_i$ (since $n_i\in\dom(u_{i+1})$) we must have $n<n_i$. By $u_i \sqsubset u_{i+1}$ we have $n\notin\dom(u_i)$. Moreover, by our main case assumption that $n_i<|s_{u,i}|$ we can assume that $i>0$ (else we would have $|s_{u,i}|=0$), and thus by definition $|s_{u,i}|=n_{i-1}+1$, and so $n_i\leq n_{i-1}$. By the assumption $S_\varphi(u)$ this can be strengthened to $n_i<n_{i-1}$. Therefore we also have $n<|s_{u,i}|$ and by the induction hypothesis obtain that $\pi_1((s_{u,i})_n)=f^k_n$ where $k$ is the least such that $\forall j \in \{k, \ldots, i-1\} (n_j>n)$. Since $n_i > n$, we have $\pi_1((s_{u,i+1})_n)=\pi_1((s_{u,i})_n)=f^k_n$, for the same $k$, which moreover satisfies the stronger property $\forall j \in \{ k, \ldots, i \}(n_j>n)$. \\
For the second case, suppose again that $n\notin\dom(u_{i+1})$ and $n<|s_{u,i+1}|=n_i$, which as in the first case we can strengthen to $n<n_i$. As in the first case we can also infer that $n\notin\dom(u_i)$. There are two subcases to deal with. Either $n\leq n_{i-1}$ and hence $n\leq|s_{u,i}|$, and so by the induction hyposthesis we have $\pi_1((s_{u,i+1})_n)=\pi_1((s_{u,i})_n)=f^k_n$ where $k$ is the least such that $\forall j \in \{k, \ldots, i-1\} (n_j>n)$. Then by the main case assumption we obtain $n_{i-1}+1\leq n_i$ and hence $n<n_{i-1}<n_i$, and thus $k$ satisfies the stronger property $\forall j \in \{ k, \ldots, i \} (n_j>n)$. Or, in the second subcase $n\geq n_{i-1}$ we have $\pi_1((s_{u,i+1})_n)=f^i_n$. We clearly have $\forall j \in \{ k, \ldots, i \} (n_j>n)$ for $i=k$ since this reduces to $n_i>n$ which we have already established. To see that $k=i$ is the least such $k$, observe that for any $k<i$ we would need $n_{i-1}>n$, which contradicts the premise of the second subcase. \end{proof}

\begin{lemma} \label{lemma-E} If $n \not\in \dom(u)$ and $n < |s_u|$ then $\pi_1((s_u)_n) = f^k_n$, where $k$ is the least such that $\forall j \in \{ k, \ldots, |\dom(u)|-1\} (n_j > n)$.
\end{lemma}
\begin{proof} This is a direct corollary of Lemma \ref{lemma-Ep}, setting $i=|\dom(u)|$ and using $u_{|\dom(u)|}=u$ (Lemma \ref{lem-S}), and recalling that $s_u:=s_{u,|\dom(u)|}$.
\end{proof}

\begin{lemma} \label{lemma-F-a} Assuming $\tilde\varphi(\wext{s_u})\geq |s_u|$ we have
\begin{equation*}
\BR_{Y,R}^{\tilde\phi,\tilde b,\tilde \varphi}(s_u) = 
\begin{cases}b(u) & \mbox{if $\varphi(\ext{u})\in\dom(u)$} \\
\phi_u(\lambda x . \BR_{Y,R}^{\tilde\phi,\tilde b,\tilde \varphi}(s_{\upd{u}{\varphi(\ext{u})}{x}})) & \mbox{otherwise}.
\end{cases}
\end{equation*}
\end{lemma}

\begin{proof} Suppose that $\tilde\varphi(\wext{s_u})\geq |s_u|$ (and $S_\varphi(u)$ as in the previous results). Define the sequences $t_{s,m}$, for $m \geq |s_u|$, by course-of-values recursion as
\begin{equation*} t_{s,m} :=_{Y^\ast} s_u \ast \pair{u,g_{|s_u|}}\ast\ldots\ast\pair{u,g_{m-1}}\end{equation*}
where $g_{k}$ is inductively defined as
\begin{equation*}
g_{k}:=
\begin{cases}
\lambda w,x . \BR(t_{s,k}\ast\pair{\initSeg{u}{k}\ast x\at w, \zero}) & \mbox{if $k\notin\dom(u)$}\\
\zero & \mbox{otherwise}.
\end{cases}
\end{equation*}
As before, let $n = \tilde\varphi(\wext{s_u}) = \varphi(\ext{u})$. We prove by induction on $m$ that $\BR(s_u)=\BR(t_{s,m})$ for $|s_u|\leq m\leq n+1$. This is true by definition for $m=|s_u|$, so assuming it is also true for arbitrary $m<n+1$ we have
\[
\begin{array}{lcl}
\BR(s_u)	&\stackrel{(\textup{IH})}{=} & \BR(t_{s,m}) \\[1mm]
		&\stackrel{(\ast)}{=} & \tilde\phi_{t_{s,m}}(\lambda z\; . \;\BR(t_{s,m}\ast z)) \\[1mm]
		&\stackrel{\scriptsize\mbox{def. $\tilde\phi$}}{=} &
			\begin{cases}
				\BR(t_{s,m}\ast \pair{d(t_{s,m}),\zero}) & \mbox{$m\in\dom(d(t_{s,m}))$} \\ 
				\BR(t_{s,m}\ast \pair{d(t_{s,m}),\lambda w,x\;. \; \BR(t_{s,m}\ast \pair{\initSeg{d(t_{s,m})}{m}\ast x\at w,\zero})})) & \mbox{otherwise}
			\end{cases}\\[5mm]
		&\stackrel{(\ast\ast)}{=} &
			\begin{cases}
				\BR(t_{s,m} \ast \pair{u,\zero}) & \mbox{$m\in\dom(u)$} \\ 
				\BR(t_{s,m} \ast \pair{u,\lambda w,x\;. \; \BR(t_{s,m} \ast \pair{\initSeg{u}{m}\ast x\at w,\zero})})) & \mbox{otherwise}
			\end{cases}\\[3mm]
		& \stackrel{\scriptsize\mbox{def. $g_m$}}{=} & \BR(t_{s,m} \ast\pair{u,g_m})\\[1mm]
		& = & \BR(t_{s,m+1}) 
\end{array}
\]
where $(\ast)$ follows from $\tilde\varphi(\wext{t_{s,m}}) = \varphi(d^\infty(\wext{t_{s,m}})) = \varphi(\ext{u})=n\geq m$, with the second equation justified by a simple argument along the lines of Lemma \ref{lemma-D} (i.e. $d^\infty$ already pick up $u$ in $s_u$, and cannot acquire any additional elements since the first component of $(\wext{t_{s,m}})_k$ for $k>|s_u|$ can only be $u$ or $\zero$). Equality $(\ast\ast)$ uses $d(t_{s,m}) = u$, which is proved similarly. By induction we have that $\BR(s_u)=\BR(t_{s,n+1})$. But
\begin{equation*}\tilde\varphi(\wext{t_{s,n+1}})=\varphi(d^\infty(\wext{t_{s,n+1}}))=\varphi(\ext{u})=n<n+1=|t_{s,n+1}|\end{equation*}
and therefore
\begin{equation*}
\begin{aligned}
\BR(s_u)&=\BR(t_{s,n+1}) \\
&=\tilde b(t_{s,n+1}) \\
&\stackrel{}{=}\begin{cases}b(u) & \mbox{if $n\in\dom(u)$} \\ 
\phi_u(\lambda x^X . (\pi_1(t_{s,n+1})_n)(u)(x)) & \mbox{otherwise}.
\end{cases}
\end{aligned}
\end{equation*}
The second equality follows by expanding the definition of $\tilde b(t_{s,n+1})$, observing as above that $d(t_{s,n+1})=u$ and recalling that $n:=\varphi(\ext{u})$.
All that remains to show is that $(\pi_1(t_{s,n+1})_n)(u)(x)=\BR(s_{u\oplus(n,x)})$ for $n=\varphi(\ext{u})\notin\dom(u)$. This can be shown as
\[
\begin{array}{lcl}
(\pi_1(t_{s,n+1})_n)(u)(x) &=& g_n(u)(x) \\[0mm]
&\stackrel{n\notin\dom(u)}{=}& \BR(t_{s,n} \ast \pair{\initSeg{u}{n}\ast x\at u,\zero}) \\[1.5mm]
&=& \BR(t_{s,n} \ast \pair{u\oplus (n,x),\zero}) \\[1.5mm]
&=& \BR(s_{u\oplus(n,x)})
\end{array}
\]
where for the last equality we have (as observed above) $s_{u\oplus(n,x),|\dom(u)|}=s_u$, and since $n \geq |s_u|$ we have, expanding the definition of $s_{u\oplus (n,x),|\dom(u)|+1}$ and $g_m$,
\begin{equation*}s_{\upd{u}{n}{x}}:=s_{\upd{u}{n}{x},|\dom(u)|+1}=s_u\ast\pair{u,g_{|s_u|}}\ast\ldots\ast\pair{u,g_{n-1}}\ast\pair{\upd{u}{n}{x},\zero}\end{equation*}
and therefore $s_{u\oplus(n,x)}=t_{s,n} \ast \pair{u\oplus(n,x),\zero}$. 
\end{proof}

\begin{lemma} \label{lemma-F} $\sBRVar$ as defined in (\ref{eqn-Phi}) satisfies equation (\ref{sbr-var-def}).
\end{lemma}
\begin{proof} We only need to consider inputs $u$ that satisfy $S_\varphi(u)$, for which we have $\sBRVar(u) = \BR_{Y,R}^{\tilde\phi,\tilde b,\tilde \varphi}(s_u)$. Let us consider the two main cases in the definition of $\BR_{Y,R}^{\tilde\phi,\tilde b,\tilde\varphi}$: \\[1mm]
(I) If $\tilde\varphi(\wext{s_u})<|s_u|$ then $\BR_{Y,R}^{\tilde\phi,\tilde b,\tilde \varphi}(s_u) = \tilde b(s_u)$. Let $n = \varphi(\ext{u})=\tilde\varphi(\wext{s_u})$ (the last equality following from Lemma \ref{lemma-C}). Consider two further subcases:
\begin{itemize}
	\item[(Ia)] If $n \in \dom(u)$ then $\varphi(\wext{d(s_u)})\in\dom(d(s_u))$ by Lemma \ref{lemma-C}. Hence, by definition of $\tilde b$ we have $\tilde b(s_u)=b(d(s_u))=b(u)$. 
	\item[(Ib)] If $n \not\in \dom(u)$ then $\varphi(\wext{d(s_u)})\not\in\dom(d(s_u))$, again by Lemma \ref{lemma-C}, and we have
	\[
	\begin{array}{lcl}
		\tilde b(s_u) & = & \phi_u(\lambda x . (\pi_1(\wext{s_u})_n)(u)(x)) \\
				& \stackrel{(*)}{=} & \phi_u(\lambda x . \sBRVar(\upd{u}{n}{x})),
	\end{array}
	\]
	where for the first equality we use Lemma \ref{lemma-C} while step $(\ast)$ is proved as follows. By Lemma \ref{lemma-E}, noting that the premises of the lemma are precisely the assumptions of cases (I) and (Ia), we have
	\begin{equation}\label{eqn-Ib}  (\pi_1(\wext{s_u})_n)(u)(x) \stackrel{n<|s_u|}{=}(\pi_1(s_u)_n)(u)(x)= f^k_n(u)(x) \stackrel{n\notin\dom(u)}{=} \BR(\initSeg{s_{u,k+1}}{n}\ast\pair{\initSeg{v_k}{n}\ast x\at u,\zero}) \end{equation}
	where $k$ is the least index such that $\forall j \in \{k, \ldots, |\dom(u)|-1\} (n_j>n)$. Now, we have
	\begin{equation*}\initSeg{v_k}{n}=\initSeg{d(s_{u,k})}{n} \stackrel{\textup{L}\ref{lemma-B}}{=} \initSeg{u_k}{n} \stackrel{(\dagger)}{=} \initSeg{u}{n} \end{equation*}
	where the last equality follows by observing that $u$ is obtained from $u_k$ by updating the latter at points $n_k,\ldots,n_{|\dom(u)|-1}$, but since $n_j>n$ in this range it is clear that $u_k(i)=u(i)$ for $i< n$. 
	Thus $\initSeg{v_k}{n}\ast x\at u=\upd{u}{n}{x}$. In addition, using again that $n_j > n$ for all $j > k$ and examining definition of $s_{u,i}$ it is easy to see that $\initSeg{s_{u,k+1}}{n}=\initSeg{s_{u}}{n}$. Therefore
	\begin{equation*}
	\begin{aligned}
		(\pi_1(\wext{s_u})_n)(u)(x)&\stackrel{(\ref{eqn-Ib})}{=}\BR(\initSeg{s_{u}}{n}\ast\pair{\upd{u}{n}{x}, \zero})\\
		& \stackrel{(i)}{=} \BR(\initSeg{s_{\upd{u}{n}{x},|\dom(u)|}}{n}\ast\pair{\upd{u}{n}{x}, \zero})\\
		& \stackrel{(ii)}{=} \BR(s_{\upd{u}{n}{x}})\\[1mm]
		&=\sBRVar(\upd{u}{n}{x})
	\end{aligned}
	\end{equation*}
	where $(i)$ follows since the $\varphi$-thread of $\upd{u}{n}{x}$ will coincide with the $\varphi$-thread of $u$ up to point $n$, and thus by extension $s_{u,i}$ will coincide with $s_{\upd{u}{n}{x},i}$ for $i\leq|\dom(u)|$. For $(ii)$ we have
	\begin{equation}\begin{aligned}s_{\upd{u}{n}{x}}&:=s_{\upd{u}{n}{x},|\dom(u)|+1}\\
	&=\initSeg{s_{\upd{u}{n}{x},|\dom(u)|}}{n}\ast\pair{\upd{u}{n}{x},\zero}	\end{aligned}\end{equation}
	where in the second equality we expanded the definition of $s_{\upd{u}{n}{x},|\dom(u)|+1}$, observing that $s_{\upd{u}{n}{x},|\dom(u)|}=s_{u,|\dom(u)|}=s_u$ and $$\ext\varphi(\initSegs{\varphi}{\upd{u}{n}{x}}{|\dom(u)|})=\ext{\varphi}(u)=n<|s_{u}|=|s_{\upd{u}{n}{x},|\dom(u)|}|.$$ For the last equation we used the simple fact that $S_{\varphi}(u)\Rightarrow S_{\varphi}(\upd{u}{n}{x})$ for $n=\varphi(\ext{u})\notin\dom(u)$.
\end{itemize}
\noindent (II) Assuming $\tilde\varphi(\wext{s_u})\geq |s_u|$ we have, by Lemma \ref{lemma-F-a}, that $\sBRVar(u) := \BR_{Y,R}^{\tilde\phi,\tilde b,\tilde \varphi}(s_u)$ satisfies equation (\ref{sbr-var-def}), again using that $S_\varphi(u)\Rightarrow S_\varphi(\upd{u}{n}{x})$ for $n\notin\dom(u)$, and thus $\sBRVar(\upd{u}{n}{x})=\BR(s_{\upd{u}{n}{x}})$. \\[1mm]
Putting both cases together we have that $\sBRVar$ as defined in (\ref{eqn-Phi}) satisfies equation (\ref{sbr-var-def}).
\end{proof}

All that remains to be shown is that the restricted version $\sBRVar$ is equivalent to the full version.

\begin{lemma}\label{lem-Psi-sBR}$\sBR$ is primitive recursively definable from $\sBRVar$, provably in $\EHAomega+(\sBRVar)+\sBI$.\end{lemma}

\begin{proof}For some arbitrary input $v$ (not necessarily satisfying $S_\varphi(v)$) we define $\sBR^{\phi,b,\varphi}(v):=\sBRVar^{\phi^v,b^v,\varphi^v}(\emptyset)$ where
\begin{equation*}\begin{aligned}\phi^v_u(p^{X\to R})&=_R\begin{cases}b(v \at u) & \mbox{if $\varphi(v\at\ext{u})\in\dom(v)$}\\
 \phi_{v\at u}(p) & \mbox{otherwise}\end{cases}\\
b^v(u)&=_R b(v\at u)\\
\varphi^v(\alpha)&=_\NN \varphi(v\at\alpha).\end{aligned}\end{equation*}
We prove that $\sBRVar^{\phi^v,b^v,\varphi^v}(\emptyset)=\sBR^{\phi,b,\varphi}(v)$ by $\sBI$ on the predicate $$P(u)\equiv [\sBRVar(u)=\sBR(v\at u)]$$ relative to the functional $\varphi^v$. First observe that $\sBRVar$ is sufficient to define the bound $\theta_{\varphi,\alpha}(\emptyset)$ of Lemma \ref{prop-sbar}, since $\theta_{\varphi,\alpha}(\emptyset)$ only makes recursive calls on $\varphi$-threads. Therefore provably in $\EHAomega + (\sBRVar)$, for all $\alpha\colon X^\NN$ there exists a least $n$ such that $\wext{\varphi^v}(\initSegs{\varphi^v}{\alpha}{n})=\ext{\varphi}(v\at\initSegs{\varphi^v}{\alpha}{n}) \in\dom(\initSegs{\varphi^v}{\alpha}{n})$. In this case $\initSegs{\varphi^v}{\alpha}{n}\in S_{\varphi^v}$ and therefore
\begin{equation*}\sBRVar(\initSegs{\varphi^v}{\alpha}{n})=b^v(\initSegs{\varphi^v}{\alpha}{n})=b(v\at\initSegs{\varphi^v}{\alpha}{n}).\end{equation*}
But since $\wext{\varphi^v}(\initSegs{\varphi^v}{\alpha}{n}) \in\dom(\initSegs{\varphi^v}{\alpha}{n})$ implies $\wext{\varphi^v}(\initSegs{\varphi^v}{\alpha}{n})\in\dom(v\at\initSegs{\varphi^v}{\alpha}{n})$ we also have $\sBR(v\at\initSegs{\varphi^v}{\alpha}{n})=b(v\at\initSegs{\varphi^v}{\alpha}{n})$.
For the induction step, for any $u\in S_{\varphi^v}$ we have, expanding $\sBRVar$ and its parameters:
\begin{equation*}\begin{aligned}\sBRVar(u)&=\begin{cases}b^v(\ext{u}) & \mbox{if $\varphi^v(u)\in\dom(u)$}\\
\phi^v_u(\lambda x\; . \; \sBRVar(u\oplus (\varphi^v(\ext{u}),x)) & \mbox{otherwise}\end{cases}\\
&=\begin{cases}b(v\at u) & \mbox{if $\varphi(v\at \ext{u})\in\dom(u)$}\\
b(v\at u) & \mbox{if $\varphi(v\at \ext{u})\in\dom(v)$}\\
\phi_{v\at u}(\lambda x\; . \; \sBRVar(u\oplus (\varphi(v\at\ext{u}),x)) & \mbox{otherwise}\end{cases}\\
&=\begin{cases}b(v\at u) & \mbox{if $\ext\varphi(v\at u)\in\dom(v\at u)$}\\
\phi_{v\at u}(\lambda x\; . \; \sBRVar(u\oplus (\ext\varphi(v\at u),x)) & \mbox{otherwise}\end{cases}\end{aligned}\end{equation*}
In the second case, setting $m=\ext{\varphi}(v\at u)=\varphi^v(\ext{u})$, since $u\oplus (m,x)\in S_{\varphi^v}$ and $m\notin\dom(u)$ we can assume as induction hypothesis $\forall x P(u\oplus (m,x))$ and thus 
\begin{equation*}
\begin{aligned}
\sBRVar(u\oplus (m,x))&=\sBR(v\at [u\oplus (m,x)]) \\
&= \sBR([v\at u]\oplus (m,x))
\end{aligned}
\end{equation*}
the first equality following from the bar induction hypothesis, and the last from $m\notin\dom(v)$, and hence $\sBRVar(u)=\sBR(v\at u)$. This establishes the premise of $\sBI$, from which we obtain $P(\emptyset)$, which completes the proof.\end{proof}

We are now able to prove the main result of the section.

\begin{proof}[Proof of Theorem \ref{thm-BR-sBR}] By Lemma \ref{lem-BR-sBR} a term satisfying the defining equations of $\sBRVar$ exists in $\EHAomega+\BR$, therefore by Lemma \ref{lem-Psi-sBR}, $\sBR$ is primitive recursively definable in turn from $\BR$, provably in $\EHAomega+\BR+\sBI$.\end{proof}

\section{Final remarks}
\label{sec-conc}

We conclude this paper by discussing, on a very informal level, how our move from $\BR$ to $\sBR$ opens the door to extending bar recursion to more complex types, thus leading to a computational interpretation of a wider class of choice principles. The following section should be considered as nothing more than a quick sketch which outlines the idea, as we do not provide any form of rigorous proof. We leave a full development of the section to future work.

\subsection{Discrete choice}

Up until now, we have considered bar recursion over either finite sequences or finite partial functions: in other words over objects of type $\NN\to X+\unit$ with finite support. It is natural ask whether we can further generalise bar recursion to take as input objects of type $D\to X+\unit$ with finite support, for some suitable class of indexing types $D$.

It is clear that for such a bar recursor to be well-defined we require equality on $D$ to be decidable, and moreover for well-foundedness we require that the stopping condition $\varphi(\ext{u})\in\dom(u)$ is eventually reached for $u$ with sufficiently large domain. This first condition is already highly restrictive: in $\PAomega$ decidability of equality is only guaranteed for types of lowest level. However, it has been shown by Escard{\'o} \cite{Escardo(2008.0)} that in the Kleene-Kreisel continuous functionals, decidability of equality can be established for a wider class of types.

\begin{definition}[Escard{\'o} \cite{Escardo(2008.0)}]Inductively define the \emph{discrete} and \emph{compact} types by
\begin{equation*}\begin{aligned}\discrete \; &::= \; \BB \; | \; \NN \; | \; \discrete\times\discrete \; | \; \compact\to\discrete\\
\compact \; &::= \; \BB \; | \; \compact\times\compact \; | \; \discrete\to\compact. \end{aligned}\end{equation*}\end{definition}
This terminology is based on the fact that the space $C_X$ of Kleene-Kreisel continuous functionals of type $X$ is discrete under the usual (quotient of the) Scott topology whenever $X$ is a discrete type, and is compact whenever $X$ is a compact type. Several properties of discrete and compact types are established in \cite{Escardo(2008.0)}, including the fact that for \emph{arbitrary} discrete $X$ the space $C_X$ is both computably enumerable and has decidable equality\footnote{However, equality may not be \emph{primitive recursively} decidable as in $\PAomega$: for non-trivial discrete types one must appeal to the so-called \emph{infinite product of selection functions} (see \cite{Escardo(2008.0)} for details).} (this is striking given that the discrete types contain genuine higher-types such as $\BB^\NN\to\NN$). Moreover, by a standard topological argument one can extend the usual sequential continuity property for functionals on infinite sequences to the following principle:
\begin{equation*}
\CONT[D] \ \colon \ \forall \varphi^{X^D\to D},\alpha^{X^D}\exists S\subset D\forall \beta(\forall d \in S(\alpha(d) =_{X} \beta(d)) \to \varphi(\alpha)=_{D}\varphi(\beta)),
\end{equation*}
where $S$ is a finite subset of $D$. 



Let $X^{\dagger D}$ denote the type of finite partial functions from $D$ to $X$, i.e. partial functions $u \colon D \to X$ with finite domain. We define $\sBR[D]$ where $D$ ranges over discrete types by
\begin{equation*}
\sBR[D]_{X,R}^{\phi,b,\varphi}(u^{X^{\dagger D}})=_R
\begin{cases}
	b(u) & \mbox{if $\ext{\varphi}(u)\in\dom(u)$} \\
	\phi_u(\lambda x^{}\; . \; \sBR[D]^{\phi,b,\varphi}(\upd{u}{\varphi(\ext{u})}{x})) & \mbox{otherwise}
\end{cases}
\end{equation*}
where $u \colon X^{\dagger D}$ and $\varphi \colon (D \to X) \to D$. Note that $\sBR$ as defined in Section \ref{sec-bar} is just $\sBR[\NN]$. This generalised form of bar recursion is well-defined in the continuous functionals since equality on $D$ is decidable, and therefore the constructions $(\hat{\cdot}) \colon X^{\dagger D} \to X^D$ and $\oplus\colon X^{\dagger D}\to (D\times X)\to X^{\dagger D}$ are still continuous (which would not be the case for e.g. $D=\NN\to\NN$).

Moreover, the recursor $\sBR[D]$ is well founded by $\CONT[D]$: Suppose that $\sBR(u)$ is not total for some total input $u$. Then by classical dependent choice we can construct a sequence recursively by $u_0:=u$ and $u_{i+1}:=u_i\oplus (d_i,x_i)$, where for each $i$ we have
\begin{equation*}
(i) \; \; d_i=\ext{\varphi}(u_i)\notin\dom(u_i) \quad \quad \quad
(ii)\; \; \sBR(u_i)\mbox{ is not total}.
\end{equation*}
By classical countable choice define $\alpha\colon D\to X+\unit$ by $\alpha(d):= x_i$ if $d=d_i$ for some $i$, and undefined otherwise. Then by $\CONT[D]$ there exists a finite subset $S\subset D$ such that $\varphi(\alpha)=\varphi(\beta)$ whenever $\alpha(d)=\beta(d)$ for all $d\in S$. Now since $\alpha$ is the domain-theoretic limit of the $u_i$, there is some index $I$ such that $u_I(d) = \alpha(d)$ for $d \in S$, and therefore $d_I=\ext{\varphi}(u_I)=\varphi(\alpha)$. Now by definition we have $d_I\in\dom(u_{I+1})$, and so in particular $d_I\in\dom(\alpha)$. It is clear that $d_I\notin S$, since by $(i)$ we have $d_I\notin\dom(u_I)$ but $d_I\in\dom(\alpha)$, contradicting the definition of $I$. But then $u_{I+1}=u_I$ on $S$, and therefore $d_{I+1}=\varphi(\alpha)=d_I$ and hence $d_{I+1}\in\dom(u_{I+1})$, a contradiction.

This constitutes an informal argument that $\sBR[D]$ is a well-defined, total continuous functional, and so by an entirely analogous procedure to the case of $\sBR[\NN]$, one can construct $f$ and $p$ in $\sBR[D]$ which satisfy the appropriate generalisation of Spector's equations:
\begin{equation*} \label{eqn-spector-discrete}
\begin{aligned}
	\varphi f& =_D d \\
	f(\varphi f)&=_X \varepsilon_d p \\
	qf&=_Y p(\varepsilon_d p).
\end{aligned}
\end{equation*}
As a result, we gain a computational interpretation of the following axiom of \emph{discrete choice}:
\begin{equation*}
\AC_{D,X} \ \colon \ \forall d^D\exists x^X A(d, x)\to\exists f^{D \to X}\forall d A(d, f d).
\end{equation*}
We remark that, as shown in \cite{Escardo(2008.0)}, the set $C_D$ is recursively enumerable for any discrete type $D$, and so can be encoded in the usual type of natural numbers $\NN$. Therefore in theory we could have defined both $\sBR[D]$ and the analogous generalisation $\BR[D]$ of Spector's bar recursion (where the points of $D$ are ordered relative to their encoding in $\NN$) in terms of $\sBR[\NN]$ and $\BR[\NN]$ respectively. However, this reduction to the base level would of course rely explicitly on the encoding of $C_D$ into $\NN$ on the meta-level. To avoid this and define the generalised recursor directly seems only possible for the symmetric recursor $\sBR[D]$, as Spector's bar recursion relies inherently on the underlying ordering of the natural numbers, and is therefore prime-facie undefined for higher level discrete types on which no natural total ordering exists. 
\subsection{Conclusion}

We have introduced a variant of bar recursion that, unlike Spector's bar recursion, carries out recursive calls in symmetric manner, as dictated by the control parameter. We have shown that this symmetric recursor exists in the usual models of bar recursion, such as the Kleene-Kreisel continuous functionals, and is in fact primitive recursively equivalent to Spector's bar recursion. We then showed that Spector's equations, a solution to which is sufficient to give a Dialectica interpretation to the negated axiom of countable choice, can be solved with a special case of symmetric bar recursion, analogously to Spector's original bar recursive solution. We compared concrete realizers obtained from the classical proof that there is no injection from $\NN^\NN\to\NN$ using both variants of bar recursion, and demonstrated that our new method of extracting programs from proofs in classical analysis performs dramatically more efficiently in many cases.

Our work fits in to the much broader program of adapting and refining traditional proof theoretic techniques so that they are better suited to their role in modern proof theory - in our case taking a well-known method of proving the consistency of classical analysis and altering it so that it becomes better suited as a tool for extracting programs from proofs. However, concrete applications of program extraction using symmetric bar recursion are currently restricted to the single case study given here, and we believe that it deserves more attention in the future. 

In addition to the extension of bar recursion to discrete choice principles sketched above, it would be particularly interesting to investigate the procedural behaviour of symmetric bar recursion, which seems much more natural than Spector's bar recursion and has close links to both the update procedures of Avigad \cite{Avigad(2002.0)} and the learning-based realizers of Aschieri et al. \cite{AscBer(2010.0)}. Making this relationship precise could lead to a much better understanding of the relationship between proof interpretations like the Dialectica intepretation and more direct learning-based interpretations of classical logic. 

\appendix

\section{Omitted proofs}
\label{app-A}

\begin{proof}[Proof of Lemma \ref{lemma-B}] By induction on $i$. If $i = 0$ the result is trivial, since $d(s_{u,0})=d(\pair{})=\emptyset$. The induction step is not much more difficult, but involves a deal of tedious verification to perform rigourously. Suppose the lemma holds for some $i<|\dom(u)|$.  There are two main cases to deal with: either $n_i<|s_{u,i}|$ or $n_i\geq |s_{u,i}|$. 

In the first case we have $s_{u,i+1}=\initSeg{s_{u,i}}{n_i}\ast\pair{u_{i+1},\zero}$. Then for $j<n_i$ we have by definition
\begin{equation*}\begin{aligned}d(s_{u,i+1})(j)&=(\pi_0(s_{u,i+1})_{i'})(j) \mbox{ for least $i'\leq j$ such that $j\in\dom(\pi_0(s_{u,i+1})_{i'})$, else undefined}\\
&=(\pi_0(s_{u,i})_{i'})(j) \mbox{ for least $i'\leq j$ such that $j\in\dom(\pi_0(s_{u,i})_{i'})$, else undefined}\\
&=d(s_{u,i})(j)\\
&\stackrel{I.H.}{=}u_i(j)\\
&=u_{i+1}(j)\end{aligned}\end{equation*}
where we use the fact that $(s_{u,i+1})_{i'}=(s_{u,i})_{i'}$ for $i'\leq j<n_i$, and in the last equality that $u_i(j)=u_{i+1}(j)$ fo $j<n_i$. Using similar reasoning, for $j\geq n_i$ we have
\begin{equation*}\begin{aligned}d(s_{u,i+1})(j)&=\begin{cases}(\pi_0(s_{u,i+1})_{i'})(j) \mbox{ for least $i'<n_i$ such that $j\in\dom(\pi_0(s_{u,i+1})_{i'})$}\\
(\pi_0(s_{u,i+1})_{n_i})(j) \mbox{ if $j\in\dom(\pi_0(s_{u,i+1})_{n_i})$}\\
\mbox{else undefined}\end{cases}\\
&\stackrel{I.H.}{=}\begin{cases}u_i(j) \mbox{ for least $i'<n_i$ such that $j\in\dom(\pi_0(s_{u,i})_{i'})$}\\
u_{i+1}(j) \mbox{ if $j\in\dom(u_{i+1})$}\\
\mbox{else undefined}\end{cases}\\
&=u_{i+1}(j).\end{aligned}\end{equation*}  
For the last equality we use the fact that $j\in\dom(u_i)$ implies $j\in\dom(u_{i+1})$ and $u_i(j)=u_{i+1}(j)$.

For the second main case $n_i\geq |s_{u,i}|$ we have $s_{u,i+1}=s_{u,i}\ast\pair{d(s_{u,i}),f^i_{|s_{u,i}|}}\ast\ldots\ast\pair{d(s_{u,i}),f^i_{n_i-1}}\ast\pair{u_{i+1},\zero}$, and the argument proceeds as in the first case: for $j<n_i$ we have (expanding the definition of $d$)
\begin{equation*}\begin{aligned}d(s_{u,i+1})&=\begin{cases}(\pi_0(s_{u,i})_{i'})(j)\mbox{ for least $i'\leq j$ such that $i'<|s_{u,i}|$ and $j\in\dom(\pi_0(s_{u,i})_{i'})$}\\
d(s_{u,i})(j)\mbox{ if $|s_{u,i}|\leq j$ and $j\in\dom(d(s_{u,i}))$}\\
\mbox{else undefined}\end{cases}\\
&\stackrel{I.H.}{=}\begin{cases}u_i(j)\mbox{ for least $i'\leq j$ such that $i'<|s_{u,i}|$ and $j\in\dom(\pi_0(s_{u,i})_{i'})$}\\
u_i(j)\mbox{ if $|s_{u,i}|\leq j$ and $j\in\dom(u_i)$}\\
\mbox{else undefined}\end{cases}\\
&=u_i(j)\\
&=u_{i+1}(j).\end{aligned}\end{equation*}
where for the last equality we use that $u_i(j)=u_{i+1}(j)$ for $j\leq n_i$. For $j=n_i$, it is easy to see that $d(s_{u,i+1})(n_i)=u_{i+1}(n_i)$, since by our assumption $S_\varphi(u)$ we know that $n_i\notin\dom(u_i)$ and thus the least $i'\leq n_i$ such that $n_i\in\dom(\pi_0(u_{u,i+1})_{i'})$ is $i'=n_i$.\end{proof}

\bibliographystyle{plain}

\bibliography{tp}

\end{document}